\documentclass[pdflatex,sn-mathphys-num]{sn-jnl}

\usepackage{graphicx}%
\usepackage{multirow}%
\usepackage{amsmath,amssymb,amsfonts}%
\usepackage{amsthm}%
\usepackage{mathrsfs}%
\usepackage[title]{appendix}%
\usepackage{xcolor}%
\usepackage{textcomp}%
\usepackage{manyfoot}%
\usepackage{booktabs}%
\usepackage{algorithm}%
\usepackage{algorithmicx}%
\usepackage{algpseudocode}%
\usepackage{listings}%
\usepackage{graphicx}

\theoremstyle{thmstyleone}%
\newtheorem{theorem}{Theorem}
\newtheorem{lemma}{Lemma}
\theoremstyle{thmstyletwo}%
\newtheorem{example}{Example}%
\newtheorem{remark}{Remark}%

\theoremstyle{thmstylethree}%
\newtheorem{definition}{Definition}%
\newtheorem{corollary}{Corollary}

\newcommand{\Ric}{\mbox{Ric}}

\newcommand{\Rscal}{\mbox{R}}
\newcommand{\Hess}{\mathrm{Hess}}
\newcommand{\Lie}{\mathcal{L}}
\newcommand{\trace}{\mbox{tr}}
\newcommand{\divergence}{\mbox{div}}

\newcommand{\grad}{\mbox{grad}}
\newcommand{\spacetime}{(L,\gamma)}

\newcommand{\idsmatter}{(M,g,K,\rho,J)}
\newcommand{\differential}{\text{d}}
\newcommand{\Hyp}{\mathbb{H}}
\newcommand{\Sph}{\mathbb{S}}
\newcommand{\Real}{\mathbb{R}}
\newcommand{\supp}{\mbox{supp}}

\begin{document}

\title[Article Title]{Dynamics of Geometric Invariants in the Asymptotically Hyperboloidal Setting: Energy and Linear Momentum}


\author*[1]{\fnm{Anna} \sur{Sancassani}}\email{sancassani@math.uni-tuebingen.de}
\equalcont{These authors contributed equally to this work.}

\author*[1]{\fnm{Saradha} \sur{Senthil Velu}}\email{saradha@math.uni-tuebingen.de}
\equalcont{These authors contributed equally to this work.}

\affil*[1]{\orgdiv{Department of Mathematics}, \orgname{University of Tübingen}, \\ \orgaddress{\street{Auf der Morgenstelle 10},  \postcode{72076}, \state{Tübingen}, \country{Germany}}}

\abstract{We investigate the evolution of geometric invariants, as defined by Michel \cite{Michel}, in the context of asymptotically hyperboloidal initial data sets. Our focus lies on the charges of energy and linear momentum, and we study their behavior under the Einstein evolution equations. We construct foliations describing the evolution of asymptotically hyperboloidal initial data sets using hyperboloidal time function. We define E-P chargeability as a property of the initial data set, and we show that it is preserved under the evolution for our choice of time function. This ensures that the charges are well-defined along the evolution, which is crucial for our approach. Along such foliations, we recover the same energy-loss and linear momentum-loss formulae as those derived by Bondi, Sachs, and Metzner \cite{Bondi-vanderBurg-Metzner} while operating under weaker asymptotic assumptions. Our approach is distinct from previous work as we do not utilize conformal compactifications and work directly at the level of the initial data set. 
}


\maketitle
\newpage

\section{Introduction}\label{introduction}
In General Relativity, observers are modeled by spacelike Cauchy hypersurfaces embedded in a spacetime. Under certain decay conditions at infinity, we can define physical quantities measured by these observers 
depending on the geometry at infinity of the hypersurface. In this paper, we are interested in studying the evolution of these quantities within the spacetime. 

Gravitational radiation propagates through the spacetime and eventually reaches null infinity. Certain observers are better suited for analyzing this radiation in terms of the energy flux at future null infinity (\(\mathscr{I}^+\)). In this paper, we consider observers modeled by asymptotically hyperboloidal initial data sets.

Unlike asymptotically Euclidean hypersurfaces, which extend toward spatial infinity, asymptotically hyperboloidal initial data remain spacelike while reaching future null infinity \(\mathscr{I}^+\) (see, e.g., \cite{Friedrich1983, AnilZenginolu2025}). In order to study the evolution of asymptotically hyperboloidal initial data sets, we adopt the concept of hyperboloidal foliations of the Minkowski spacetime \cite{Anil, AnilZenginolu2025}. These provide a natural foliation of $\mathscr{I}^{+}$ and define time functions which remain regular as one approaches infinity. This is necessary to avoid the phenomenon known as ``freezing of time", in which the time function degenerates at $\mathscr{I}^+$, as seen in the Milne foliation \cite{Anil, AnilZenginolu2025}, or at spatial infinity, as in the case of asymptotically Euclidean initial data.  

Energy and linear momentum on asymptotically hyperboloidal initial data sets have been studied extensively in \cite{JacekJezierski2002, Trautman, Chruciel2004TheTM, Chrusciel1998}, typically derived via the Hamiltonian formalism. In this paper, we instead adopt the formulation developed by Michel \cite{Michel}, where energy and linear momentum are defined as geometric invariants intrinsic to the initial data. These geometric invariants are defined on non-compact Riemannian manifolds that approach a certain background in an asymptotic end, and they correspond to the Killing initial data of the background. The derivation in \cite{Michel} is done in an abstract and general setting. In this paper, we apply it to the case where the charge density is described by the operator of constraints in General Relativity, where the initial data set is asymptotic to the standard hyperboloid in Minkowski. After prescribing conditions on the asymptotic behavior of the geometry on the initial data that ensure integrability, we obtain a charge integral for each Killing initial data of the background hyperboloid. 

These quantities are conceptually distinct from those derived in the Hamiltonian formulation. Some differences can be noted between the charges, particularly in how contractions are taken: in the Hamiltonian approach, contractions are taken with respect to the metric of the initial data, whereas in Michel's framework, the background metric is used. We obtain the same formulae of the energy and linear momentum for our definition of asymptotically hyperboloidal initial data sets. For a detailed review of the various definitions of energy in asymptotically hyperboloidal initial data sets, we point the reader to \cite{Chruciel2004TheTM}.  

In this paper, we study the evolution of the energy and linear momentum defined by Michel directly using the Einstein evolution equations. A central result in this context is the Bondi energy loss formula, which quantifies the decrease in total energy due to gravitational wave emission \cite{Bondi-vanderBurg-Metzner, RKSachs1962, Zhang2006}. While traditionally derived within the Bondi-Sachs formalism at null infinity, the hyperboloidal approach provides an alternative formulation
of energy loss without the use of null coordinates. In this setting, we formulate the change in energy and linear momentum in terms of spacelike initial data and its evolution, offering a natural connection between the spacelike initial value problem and asymptotic radiation properties.

We note here that our result does not rely on conformal compactifications. Conformal methods are often employed in numerical relativity (see, e.g., \cite{Frauendiener1998}) to bring infinity within a finite computational domain and define structures at \(\mathscr{I}^+\). Instead, we work directly in the physical spacetime and track the contribution of individual terms in the charge integrals for energy and linear momentum. This approach ensures that our analysis remains independent of conformal rescalings.

We observe that the evolution of energy and linear momentum in asymptotically hyperboloidal initial data sets closely mirrors the mass loss behavior in the null formalism using Bondi-Sachs coordinates. However, instead of using null coordinates, we work with standard coordinates on the hyperboloid $\Hyp^3$ and an appropriate choice of time function $\tau$ defining the hyperboloidal foliation. A similar result is obtained by \cite{ChenWangYau}, who derive an energy loss formula using the Liu-Yau quasi-local mass in the large sphere limit via optimal isometric embeddings. This definition of quasi-local energy, in the large sphere limit, yields the same mass aspect as the global charge defined by Michel. 

An important observation from the evolution of asymptotically hyperboloidal initial data is that along the foliation, the asymptotics are not preserved and we find that they vary at higher orders. This is a natural effect of the hyperbolic geometry of the initial data under consideration, as opposed to the asymptotics being preserved in the evolution of asymptotically Euclidean initial data sets. We introduce a notion of \emph{E-P chargeability} to describe initial data that admit well-defined energy and linear momentum \`a la Michel. We prove that, starting from an asymptotically hyperboloidal initial data set, the property of E-P chargeability is preserved under evolution with respect to the chosen time function. This is a consequence of subtle cancellations in the evolution of the charges.

The paper is structured as follows. In Section \ref{sec2}, we introduce the spacelike initial value problem and define asymptotically hyperboloidal initial data sets. In Section \ref{sec3}, we review Michel’s definitions of energy and linear momentum and derive their expressions for the class of asymptotically hyperboloidal observers. In Section \ref{sec4}, we discuss our choice of foliation for the evolution. In Section \ref{sec5}, we prove the energy and linear momentum loss for asymptotically hyperboloidal initial data sets. Moreover, we define E-P chargeability and show that it is preserved along the evolution. In Section \ref{sec6}, we explore possible generalizations of the class of observers to which our theorems apply. 

\section*{Acknowledgements}
The authors thank Carla Cederbaum for insightful discussions and Simone Coli for help with figures. 

\section{Preliminaries}\label{sec2}
Let \(\spacetime\) be a spacetime and let $t: \Real\supseteq(a,b) \rightarrow \Real$ be a time function. The level sets of $t$ define a local foliation of the spacetime manifold by spacelike hypersurfaces $\{M_t\}_{t \in (a,b)} \subseteq L$.  The Einstein equations, together with the Gauss--Codazzi and Mainardi equations, impose constraints on the geometry of each embedded hypersurface. These constraints are encoded in the constraint operator $\Phi$, which acts on pairs $(g, K)$,  where \(g\) is a Riemannian metric and \(K\) is a symmetric $(0,2)$-tensor. Specifically, we define  
\[
\Phi(g,K):=
\left(
    \begin{matrix}
         \Rscal+(\trace_gK)^2-|K|_g^2\\
          2\divergence_g(K-\trace_gK\cdot g)
    \end{matrix}
\right),
\]
where \(\Phi\) is a map from the bundle \(\mathcal{M} \times_M S^2(T^*M)\), where $\mathcal{M}$ is the natural bundle of metrics on a manifold $M$, and $S^2(T^*M)$ is the bundle of symmetric $(0,2)$-tensors on $M$, to \(\mathcal{C}^\infty(M) \times \Gamma(T^*M)\).  

For each hypersurface \(M_t\) in the foliation, let \(\!^{(t)}g\) and \(\!^{(t)}K\) denote the metric and second fundamental form induced by the spacetime metric \(\gamma\). The constraint equations then require that  
$$\Phi\left(\!^{(t)}g, \!^{(t)}K\right) = (\rho,J), $$
where $\rho$ and $J$ are, respectively, the energy and momentum density of the spacetime measured by an observer moving along the foliation. 

Let $t_0\in(a,b)$, and let $M:=M_{t_0}$. We can define coordinates $(t,x^i)$ in a spacetime neighbourhood of a point $q\in M$ by flowing coordinates $(x^i)$ defined on an open set $U\subseteq M$ along a timelike vector field $\partial_t\in\Gamma(TL)$ with $dt(\partial_t)=1$. The spacetime metric $\gamma$ can then be written locally as
    \begin{align}\label{eq:gdecomposition}
        \gamma&=-{N}^2d{t}^2+\ \!^{(t)}g_{ij}\,(d{x^i}+{X}^id{t})(d{x^j}+{X}^jd{t}),
   \end{align}
   where $N$ and $X^{i}$ are the \emph{lapse} function and the \emph{shift} vector field of the foliation, and they are related to the choice of time function in the sense that ${\partial_t}\big|_M=N\cdot\nu+X$, where $\nu$ is the future-directed unit normal to $M$.

In vacuum\footnote{In this paper we will only need the vacuum evolution equations, as we will only consider the evolution of initial data that are vacuum outside a compact set.}, the metric and second fundamental form induced on the leaves vary along the foliation as described by the {Einstein evolution equations} \cite{Bartnik2004} 
    \begin{equation}\label{eq::Vacuum-EvoEqns}
    \begin{cases}
        \mathcal{L}_\nu \, {}^{(t)}g = 2\, {}^{(t)}K, \\
        \mathcal{L}_\nu {}^{(t)}K - \dfrac{{}^{(t)} \boldsymbol{\nabla}^2 N}{N} = 2 \left({}^{(t)}K\circ{}^{(t)}K\right) - {}^{(t)} \boldsymbol{\Ric} - (\ \!^{(t)}\operatorname{tr} {}^{(t)}K\ ) {}^{(t)}K.
    \end{cases}
\end{equation}
  
By a fundamental theorem of Choquet-Bruhat and Geroch \cite{YCB,choquetBruhat&Geroch} we know that the converse is also true. A Riemannian manifold $(M,g)$, together with a symmetric $(0,2)$-tensor $K$ satisfying the constraint equations, constitute well-defined initial data for the Cauchy problem associated with the Einstein Equation. In particular, there exists a unique maximal globally hyperbolic development where the initial data set embeds isometrically as a spacelike hypersurface. The above holds for any spacetime dimension $n+1$, but for the purposes of this paper, we will focus on the case $n=3$.

\begin{definition}
    We say that $\idsmatter$ is an \emph{initial data set} for the Einstein Equation if $M$ is a $3$-dimensional manifold, $g$ is a Riemannian metric on $M$ and $K$ a symmetric $(0,2)$-tensor on $M$, such that $\Phi(g, K)=(\rho,J)$. We say that the initial data is \emph{vacuum} if $\Phi(g, K)=0$.
\end{definition}
The globally hyperbolic development of the initial data admits a global time function. At least locally, we can introduce coordinates \((t, x)\) such that the initial hypersurface is given by \(M = \{t = 0\}\) and the spacetime foliates in a neighborhood of \(M\) as \(\{t = \delta\}\) for \(\delta \in (-\varepsilon, \varepsilon)\). The evolution equations \eqref{eq::Vacuum-EvoEqns} then describe the evolution of the induced geometry along this foliation.  

This is the approach we adopt in this paper. Starting with vacuum initial data sets that are asymptotically hyperboloidal in the sense of Definition \ref{Def::AHyperboloidal}, we compute their energy and linear momentum using Michel’s approach. We then analyze the evolution of these quantities along a foliation, defined in Section \ref{sec3}, in a neighborhood of the initial data.

\subsection{Asymptotically Hyperboloidal Initial Data Sets}\label{sec2.1}
We denote by $(\Hyp^3,b,b)$ the hyperboloid of radius $1$ embedded in Minkowski Spacetime as a totally umbilic spacelike hypersurface
\begin{align}\label{Eqn::StandardHyperboloid}
    \Hyp^3:=\{(t,x)\in\Real^{3,1}\ |\ t^2=r^2+1\},\\
    b=\frac{1}{1+{r^2}}dr^2+r^2\sigma_{\alpha\beta}du^\alpha du^\beta,
\end{align}
where $\sigma$ is the standard round metric on the unit sphere. On the hyperboloid, we use coordinates $x=(r,u^J)$, where $u^J\in\Sph^{2}$ denote the standard polar coordinates on the sphere. Throughout the paper we will use the term \emph{hyperboloidal initial data} to refer to $(\Hyp^3,b,b)$.

\begin{definition}\label{Def::AHyperboloidal}
    An initial data set $(M, g, K, \rho, J)$ is \emph{asymptotically hyperboloidal} if there exist compact sets $B\subset M$ and $B_0\subset \Hyp^3$, and a diffeomorphism at infinity
    \[
        \Psi: M\setminus B\rightarrow \Hyp^3\setminus B_0,
    \]
    such that the $(0,2)$-symmetric tensors on $\Hyp^3\setminus B_0$ 
    \[\dot{g}:=\Psi^*g-b \mbox{ and }\ p:=\Psi^*K-\Psi^*g\]
    have the following asymptotic behavior as $r\rightarrow\infty$
    \begin{subequations}\label{Eqn::AHyperboloidal(g,p)}
  \begin{align}
    \dot{g}_{rr} &= \frac{\mathbf{m}_{rr}}{r^5} + \frac{\tilde{\mathbf{m}}_{rr}}{r^{6}} + O_{2}\bigl(r^{-7}\bigr),\\[1mm]
    \dot{g}_{r\alpha} &= \frac{\mathbf{g}_{r\alpha}}{r^{3}} + O_{2}\bigl(r^{-4}\bigr),\\[1mm]
    \dot{g}_{\alpha\beta} &= {}^{0}\mathbf{g}_{\alpha\beta} + \frac{\!^g\mathbf{m}_{\alpha\beta}}{r} + \frac{{}^{g}\tilde{\mathbf{m}}_{\alpha\beta}}{r^{2}} + O_{2}\bigl(r^{-3}\bigr),
  \end{align}
  \begin{align}
    {p}_{rr} &= \frac{\mathbf{p}_{rr}}{r^{5}} + \frac{\tilde{\mathbf{p}}_{rr}}{r^{6}} + O_{1}\bigl(r^{-7}\bigr),\\[1mm]
    {p}_{r\alpha} &= \frac{\mathbf{p}_{r\alpha}}{r^{3}} + O_{1}\bigl(r^{-4}\bigr),\\[1mm]
    {p}_{\alpha\beta} &= {}^{0}\mathbf{p}_{\alpha\beta} + \frac{{}^{k}\mathbf{m}_{\alpha\beta}}{r} + \frac{{}^{k}\tilde{\mathbf{m}}_{\alpha\beta}}{r^{2}} + O_{1}\bigl(r^{-3}\bigr).
  \end{align}
\end{subequations}
Here, $\mathbf{m}_{rr}$, $\mathbf{p}_{rr}$, $\tilde{\mathbf{m}}_{rr}$ and $\tilde{\mathbf{p}}_{rr}$ are functions on $\Sph^2$ that do not depend on $r$, and 
$\mathbf{g}_{r\alpha}$ and $\mathbf{p}_{r\alpha}$ are one-forms that do not depend on $r$. The indices on these one-forms are raised and lowered with respect to $\sigma_{\alpha\beta}$.
Moreover, ${}^{0}\mathbf{g}_{\alpha\beta}$, ${}^{0}\mathbf{p}_{\alpha\beta}$, ${}^{g}\mathbf{m}_{\alpha\beta}$, ${}^{k}\mathbf{m}_{\alpha\beta}$, ${}^{g}\tilde{\mathbf{m}}_{\alpha\beta}$ and ${}^{k}\tilde{\mathbf{m}}_{\alpha\beta}$ are  symmetric $(0,2)$-tensors on $\Sph^2$ which do not depend on $r$, whose indices are raised and lowered with $\sigma_{\alpha\beta}$. Moreover, we assume that  ${}^{0}\mathbf{g}_{\alpha\beta}$, ${}^{0}\mathbf{p}_{\alpha\beta}$, ${}^{g}\tilde{\mathbf{m}}_{\alpha\beta}$ and ${}^{k}\tilde{\mathbf{m}}_{\alpha\beta}$ are traceless with respect to $\sigma$. All terms in the expansion of $\dot{g}$ are in $C^{2}(\Sph^{2})$ and those in the expansion of $p$ are in $C^{1}(\Sph^{2})$.
\end{definition}

In the above definition, we distinguish between the terms in the expansion of $\dot{g}$ and $p$ that contribute to the charges, as will be shown in Theorem \ref{THM::Energy_and_Linearmomentum_AH}. In particular, the terms contributing to the energy and linear momentum are denoted with boldface $\mathbf{m}$.

We note that the energy and momentum densities $\rho$ and $J^{i}$ in the definition above decay at the rate $O(r^{-4})$, following from the constraint equations coupled with the asymptotics of $g$ and $K$.

We also note that the definition of asymptotically hyperboloidal initial data sets can be generalized by requiring the metric and extrinsic curvature to lie in appropriate function spaces, such as weighted Hölder or weighted Sobolev spaces, for e.g., see \cite{Sakovich2021}. For the purposes of this work, however, it is sufficient to consider $C^{k}$ spaces.

\begin{remark} The definition of asymptotically hyperboloidal initial data adopted in this paper, Definition \ref{Def::AHyperboloidal}, is based on the framework established by Chen, Wang, and Yau \cite{ChenWangYau}, but there are key differences. Notably, we impose the condition that the tensors ${}^{g}\tilde{\mathbf{m}}_{\alpha\beta}$ and ${}^{k}\tilde{\mathbf{m}}_{\alpha\beta}$ be traceless with respect to the 2-metric $\sigma$ as without it, the evolution of energy and linear momentum would not be well-defined. Moreover, these conditions resemble those placed on the angular components of the expansion of the metric in Bondi--Sachs (null) coordinates (see, e.g., \cite{Godazgar2019, Chruciel2004TheTM, BMS}). These tracelessness conditions are consistent with those required for the outgoing radiation condition \cite{RKSachs1961}. Furthermore, our assumptions are less restrictive since we do not impose 
 any condition on the regularity of $\mathscr{I}^+$ (see discussion in Remark \ref{Anderssoncomparison}).
While alternative definitions of asymptotically hyperboloidal initial data exist in the literature, such as the asymptotics introduced by Wang \cite{Wang}, which have been used in \cite{Sakovich2021} to prove a positive mass theorem on asymptotically hyperboloidal initial data using the Jang equation, it is important to note that Wang's asymptotics do not incorporate the case of radiation, which is the central focus of our study. \end{remark}

\begin{remark}
The expressions of energy and linear momentum that we derive in Theorem \ref{THM::Energy_and_Linearmomentum_AH} are still valid for asymptotically hyperboloidal initial data sets as defined by Chen, Wang, and Yau and coincide with the definitions in \cite{ChenWangYau}, where it is shown that these integrals arise as limits at spatial infinity of the quasi-local Liu-Yau mass under these asymptotic conditions (c.f. Theorem 3.2 in \cite{ChenWangYau}). However, it is not explicitly shown in \cite{ChenWangYau} that this fact is preserved under the evolution since the evolved data is no longer asymptotically hyperboloidal (c.f. discussion following Theorem \ref{thm:evolution_P}). 
Instead, the additional tracelessness conditions on ${}^{g}\tilde{\mathbf{m}}_{\alpha\beta}$ and ${}^{k}\tilde{\mathbf{m}}_{\alpha\beta}$ that we impose become essential in our approach to have well-definedness of the evolution.  
\end{remark}

\begin{example}\label{example:SchwarzschildAH}
    We consider the totally umbilic initial data set 
    $(\Hyp^3_S,g_S,g_S)$,
    where $\Hyp^3_S$ is the hyperboloid in Schwarzschild spacetime of mass $m>0$. The metric $g_S$ written in standard Schwarzschild coordinates $(r,\theta,\varphi)\in (r_0,+\infty)\times\Sph^2$ is \begin{equation}\label{eq::SchwarzschildAdSmetric}
        g_S:=-\left(1-\frac{2m}{r}+r^2\right)^{-1}dr^2+r^2d\sigma^2.
    \end{equation}
    We obtain that in these coordinates, the initial data $(\Hyp^3_S,g_S,g_S)$ is asymptotically hyperboloidal in the sense of Definition \ref{Def::AHyperboloidal}. In particular, as $r\rightarrow+\infty$, we have
    \begin{align*}
            &\dot{g}_{rr} = \dot{K}_{rr}= \frac{2m}{r^5}-\frac{4m}{r^7}+O(r^{-8}),\\
            &\dot{g}_{r\alpha}=\dot{K}_{r\alpha}=0, \ \ \ \dot{g}_{\alpha\beta}=\dot{K}_{\alpha\beta}=0.
    \end{align*}
\end{example}
    
\section{Energy and linear momentum on asymptotically hyperboloidal initial data}\label{sec3}
    Our definitions of energy and linear momentum are derived from a procedure introduced by Michel \cite{Michel}, where the author introduces a method for defining geometric invariants on Riemannian manifolds that are approaching some background manifold at infinity. Here, we apply Michel's result to asymptotically hyperboloidal initial data, with the constraint operator playing the role of a charge density. We believe that providing the details of the derivation of these quantities can be helpful for the reader as the construction in \cite{Michel} involves a high level of generality and is not immediately accessible at first reading. Additionally, this allows us to motivate their interpretation as energy and linear momentum of the initial data set from a physical perspective. 

    We start from an initial data set $(M, g, K)$ which is asymptotically hyperboloidal in the sense of Definition \ref{Def::AHyperboloidal}, with respect to a diffeomorphism at infinity $\Psi$. We define 
    \[
        \dot{g}:=\Psi^*g-b, \mbox{ and } \dot{K}:=\Psi^*K-b,
    \]
    where $b$ denotes the hyperboloidal metric as in \eqref{Eqn::StandardHyperboloid}. 
    
    In Michel's approach \cite{Michel}, the charge density is described by a natural tensor-valued operator $\Phi$ that is invariant under diffeomorphisms. 
    We choose to study the case where $\Phi$ 
    is the constraint operator
    \begin{align*}
        \Phi:\mathcal{M}\times_MS_2M&\rightarrow \Gamma(\Real\oplus T^*M)\\
        (g, K)&\mapsto \left(\begin{matrix}
            \Rscal+(\trace_gK)^2-|K|_g^2\\
            2\divergence_g(K-\trace_gK\cdot g)
        \end{matrix}\right)=:\left(\begin{matrix}
            \Phi^H(g,K)\\
            \Phi^M(g,K)
        \end{matrix}\right),
    \end{align*}
    where $\mathcal{M}$ is the bundle of Riemannian metrics on $M$ and $S_2M$ are the symmetric $(0,2)$-tensors on $M$. This is an admissible choice for the charge density operator, as the condition that $\Phi$ is invariant under diffeomorphisms translates to $\Phi^H = \mbox{const}$ and $\Phi^M = 0$, meaning the background data should be vacuum (or cosmological). This requirement is satisfied for asymptotically hyperboloidal initial data sets, which approach the vacuum initial data $(\Hyp^{3}, b, b)$. We denote by $\Phi_b=(0,0)$ the constraint operator evaluated at the background $(\Hyp^3, b, b)$, and by $D\Phi_b$ its linearization at $(\Hyp^{3}, b, b)$:
    \begin{equation}
        D\Phi_b(\dot{g},\dot{K})=
        \left(
            \begin{matrix}
                 \divergence_b\divergence_b\dot{g}+\Delta_b\trace_b\dot{g}-2\trace_b\dot{g}+4\trace_b\dot{K}
                  \\
                 d\trace_b\dot{g}-b^{ij}\nabla\dot{g}_{ij}-2(\divergence_b\dot{g}-d\trace_b\dot{g})+2\divergence_b\dot{K}-d\trace_b\dot{K})
          \end{matrix}
        \right).
    \end{equation}
        
    Now, for fixed $(\dot{g},\dot{K})$, and after Taylor expansion, contracting with a test-function $\mathcal{V}\in\Gamma(\Real\oplus T^*\Hyp)$ leads to the following
    \[
        \langle \Phi(g,K)-\Phi_b , \mathcal{V} \rangle_b = 
        \langle D\Phi_b(\dot{g},\dot{K}) , \mathcal{V}\rangle_b + \mathcal{Q}(\mathcal{V},(\dot{g},\dot{K})),
    \]
    where $\mathcal{Q}(\mathcal{V},(\dot{g},\dot{K}))$ is the quadratic and higher order remainder of the expansion contracted with $\mathcal{V}$, and $\langle\cdot,\cdot\rangle_b$ denotes the inner product defined with respect to the background metric $b$. Integration by parts of the first term on the right-hand side of the above expression yields:
    \[
     \langle \Phi(g,K) , \mathcal{V} \rangle_b = 
        \divergence_b\mathbb{U}(\mathcal{V},(\dot{g},\dot{K}))+
        \langle D^*\Phi_b\mathcal{V},(\dot{g},\dot{K})\rangle_b + \mathcal{Q}(\mathcal{V},(\dot{g},\dot{K})),
    \]
    where $D\Phi_b^*$ is the formal adjoint of the linearisation of $\Phi$ at the background. The boundary term $\mathbb{U}(\mathcal{V},(\dot{g},\dot{K}))$ is a one-form of order reduced by one compared to $\Phi$ in each of its arguments, and it is defined in \cite{Michel} as the \emph{charge integrand}.
    
    This procedure leads to the construction of a geometric invariant on the initial data set $(M, g, K)$ defined with respect to a diffeomorphism at infinity $\Psi$, for each non-trivial $\mathcal{V}\in\Gamma(\Real\oplus T^*\Hyp)$, such that $D^*\Phi_b\mathcal{V}=0$ and the following integral converges
    \begin{equation}\label{Eq::MichelCharge}
        m(\dot{g},\dot{K},\mathcal{V}):=\lim_{k\rightarrow\infty}\frac{1}{16\pi}\oint_{\mathcal{S}_k}\mathbb{U}(\mathcal{V},(\dot{g},\dot{K}))(\nu)dS.
    \end{equation}
    Each $S_k=\partial B_k$ is smooth and compact, $(B_k)_{k\in\mathbb{N}}$ is a non-decreasing exhaustion of the background manifold, and $\nu$ and $dS$ are the outer unit normal and surface measure on $S_k$ with respect to the background metric. \\

    \begin{remark}
        The above integral is independent of the chosen exhaustion but depends on the coordinates at infinity $\Psi$. However, it has been shown \cite{Michel} that it is invariant under changes of coordinates at infinity which are asymptotic to the identity. 
    \end{remark}
 Choosing $\Phi$ as the constraint operator, we connect the geometric invariants defined above to their interpretation as physical quantities. In particular, one should notice that charges of type \eqref{Eq::MichelCharge} originate from elements of the kernel of the map $D^{*}\Phi_{b}$. These are, by definition, the Killing initial data (KIDs) of the background \cite{NewKIDs} -- in our case, of the hyperboloidal initial data $(\Hyp^{3}, b, b)$  -- and they are in a one-to-one correspondence with the Killing vector fields of the associated spacetime. This connection is well-known, see for example \cite{Moncrief, Berger}. The spacetime arising as the maximal globally hyperbolic development of the hyperboloidal initial data set is "a piece of" Minkowski. 
    
Noether's theorem connects the existence of Killing vector fields (symmetries) in spacetime to "conserved quantities" of the system. In the Minkowski spacetime, the conserved quantities associated with these symmetries are well understood within the theory of special relativity \cite{Christodoulou}. In particular, the quantity associated with the time-translation symmetry represents the energy of the system, while those corresponding to the spatial translation symmetries describe the linear momentum. Motivated by this, we define the energy and linear momentum on asymptotically hyperboloidal initial data sets, as the charges corresponding to the time and spatial translation Killing initial data on the background.

    \begin{definition}[Michel Charges]\label{Def::MichelChargeAH}
        Let $\idsmatter$ be an asymptotically hyperboloidal initial data set in the sense of Definition \ref{Def::AHyperboloidal}. 
        The \emph{energy} of the initial data set is given by  
         \[
        E:=m(\dot{g},\dot{K},\mathcal{V}_0)=\lim_{k\rightarrow\infty}\frac{1}{16\pi}\oint_{\mathcal{S}_k}\mathbb{U}(\mathcal{V}_0,(\dot{g},\dot{K}))(\nu)dS,
    \]
    where $\mathcal{V}_0$ is the KID on $(\Hyp^3,b,b)$ associated to the time translation symmetry in Minkowski.
    The {linear momentum} is the $3$-vector $P$, with components
        \[
        P^i:=m(\dot{g},\dot{K},\mathcal{V}_i)=\lim_{k\rightarrow\infty}\frac{1}{16\pi}\oint_{\mathcal{S}_k}\mathbb{U}(\mathcal{V}_i,(\dot{g},\dot{K}))(\nu)dS,
    \]
    where $\mathcal{V}_i, i=1,2,3$ are the KIDs on $(\Hyp^3,b,b)$ associated with the spatial translation symmetries in Minkowski.
    In the above integrals, $S_k=\partial B_k$ are smooth and compact, and $(B_k)_{k\in\mathbb{N}}$ is a non-decreasing exhaustion of the hyperboloid $\Hyp^3$, and $\nu$ and $dS$ are the outer unit normal and surface measure on $S_k$ with respect to the background hyperboloidal metric $b$.
\end{definition}

Using the definitions introduced 
above, we now proceed to compute the energy and linear momentum for a vacuum, asymptotically hyperboloidal initial data set. These quantities are independent of the chosen exhaustion of the background data and we derive them by evaluating the relevant surface integrals on spheres $\Sph^2_r$ on the hyperboloidal background, with \( r \to \infty \). The following theorem provides the explicit expressions for the energy and linear momentum of the initial data, as defined by the surface integrals in the previous section.
\begin{theorem}\label{THM::Energy_and_Linearmomentum_AH}
    Let $\idsmatter$ be an asymptotically hyperboloidal initial data set in the sense of Definition \ref{Def::AHyperboloidal}. The energy $E$ and linear momentum $P^{i}$, for $i=1,2,3$ of the initial data are well-defined and are given by the following formulae:
    \begin{equation}\label{Eqn::EnergyAH}
        E = \frac{1}{16\pi}\int_{\Sph^{2}} \left(2\mathbf{m}_{rr} + 3 \, \trace_{\sigma}\ \!^{g}\mathbf{m}+ 2\trace_{\sigma}\ \!^{k}\mathbf{m}\right) \: dA_{\Sph^{2}},
    \end{equation}
\begin{equation}\label{Eqn::MomentumAH}
    P^{i} = \frac{1}{16\pi}\int_{\Sph^{2}}\left( 2\mathbf{m}_{rr} + 3 \, \trace_{\sigma}\ \!^{g} \mathbf{m} + 2\trace_{\sigma}\ \!^{k} \mathbf{m} \right) x^{i} \, dA_{\Sph^{2}}, \quad i = 1,2,3.
\end{equation}
where $x^{i}$ denote the first spherical harmonics on the unit sphere $\Sph^2$. \end{theorem}

\begin{proof}
    Let $\mathcal{V}=(V, Y)=(V,\alpha^\# )$ be a KID of the background hyperboloidal data. Writing \eqref{Eq::MichelCharge} explicitly, we obtain that a well-defined charge associated to $\mathcal{V}$ exists whenever the following integral converges
    \begin{align*}
         m(\dot{g},\dot{K},\mathcal{V})=\lim_{r\rightarrow\infty}\frac{1}{16\pi}\oint 
            \big[ 
                V\big(&\divergence_b\dot{g}-d\trace_b\dot{g}\big)
                -\iota_{\nabla V}\dot{g} +(\trace_b\dot{g})dV
                +2\big(\iota_\alpha\dot{K}-(\trace_b\dot{K})\alpha)\big)\\
                &+(\trace_b\dot{g}) \iota_\alpha K_{\Hyp^3} 
                +\langle K_{\Hyp^3},\dot{g}\rangle_b\alpha-2\iota_\alpha(\dot{g}\circ K_{\Hyp^3}) \big](n)
                dA_{\Sph^2_r}.
    \end{align*}
    Here, $dA_{\Sph^2_r}$ is the area element on the sphere of radius $r$ in $(\Hyp^3,b)$, and $n$ is the outward unit normal to $\Sph^2_r$ within the hyperboloidal manifold. Moreover, $\divergence_b,\trace_b, \langle\cdot,\cdot\rangle_b$ are the divergence, trace, and inner product defined via the background hyperbolic metric $b$, and $K_{\Hyp^3}=b$ is the second fundamental form of the hyperboloid.
    
    Let $\{\partial_t,\partial_r,\partial_\theta, \partial_\varphi\}$ be the frame associated to the spacetime coordinates $\{t,r,\theta,\varphi\}$ in the Minkowski spacetime $\left(\Real^{3,1},\eta\right)$. We define a frame on $T\Hyp^3$, given by $\partial_r^{\mbox{\footnotesize{tang}}}, \partial_\theta, \partial_\varphi$, where $\partial_\theta, \partial_\varphi$ are the restrictions to the hyperboloid of the coordinate vector fields in Minkowski, and $\partial_r^{\mbox{\footnotesize{tang}}}$ is defined as a combination of ${\partial_t}\big|_{\Hyp^3}$ and ${\partial_r}\big|_{\Hyp^3}$ that is tangential to the hyperboloid and such that $dr(\partial_r^{\mbox{\footnotesize{tang}}})=1$. We obtain: 
    \[
        \partial_r^{\mbox{\footnotesize{tang}}}=\frac{-r^2}{\sqrt{1+r^2}}{\partial_t}\big|_{\Hyp^3}+{\partial_r}\big|_{\Hyp^3}.
    \]
    Then, substituting
    \[
        n=\frac{\partial_r^{\mbox{\footnotesize{tang}}}}{| \partial_r^{\mbox{\footnotesize{tang}}}|_{\eta}}=-r{\partial_t}\big|_{\Hyp^3}+\sqrt{1+r^2}{\partial_r}\big|_{\Hyp^3},
    \]
    we can rewrite the charge integrand as follows
    \begin{align*}
        \big[V\big(\divergence_b\dot{g}-d\trace_b\dot{g}\big)(\partial_r^{\mbox{\footnotesize{tang}}}) 
        -&(\trace_b\dot{g})b(\nabla V,\partial_r^{\mbox{\footnotesize{tang}}})
        +2(\trace_b\dot{K})b(\nabla V,\partial_r^{\mbox{\footnotesize{tang}}}) \\
        &+\dot{g}(\nabla V,\partial_r^{\mbox{\footnotesize{tang}}}) 
        -2\dot{K}(\nabla V,\partial_r^{\mbox{\footnotesize{tang}}})\big] \sqrt{1+r^2},
    \end{align*}
    where we used the fact that for KIDs on the hyperboloid, it holds $Y=-\nabla V$ if $V\not\equiv0$. 
    For ease of readability, we break down the above expression into manageable parts and compute them separately. We report here the details of the computation for the energy; the linear momentum components can be obtained following the same computation, using the corresponding KIDs.\\
    \ \\
    \emph{Energy:} Let $\dot{g},\dot{K}=\dot{g}+p$ be as in Definition \ref{Def::AHyperboloidal}, and let $\mathcal{V}_0$ be the KID on $(\Hyp^3,b,b)$ associated to the time translation symmetry in Minkowski,
    \begin{equation*}
       \begin{matrix}
           \mathcal{V}_0=\{ V=\sqrt{1+r^2},
        & Y=-r\sqrt{1+r^2}\partial_r^{\mbox{\footnotesize{tang}}} \}.
        \ & \
       \end{matrix} 
    \end{equation*}
     We compute the leading order of the terms appearing in the charge integrand, for the given asymptotic expansions of $\dot{g}$ and $\dot{K}$, 
    \begin{gather*}
    \divergence_b\dot{g}(\partial_r^{\mbox{\footnotesize{tang}}}) =
    (1+r^2)\left[\frac{\partial}{\partial r}\dot{g}_{rr}+2\left(\frac{r}{1+r^2}+\frac{1}{r}\right)\dot{g}_{rr}\right]\sim O(r^{-4}), \\
    d(\trace_b\dot{g})(\partial_r^{\mbox{\footnotesize{tang}}}) = 
    (1+r^2)\frac{\partial}{\partial r}\dot{g}_{rr}+ 2r\dot{g}_{rr}
    -\frac{2}{r^3}\sigma^{\alpha\beta}\dot{g}_{\alpha\beta}
    -\frac{1}{r^3}\sigma^{\alpha\beta}\frac{\partial}{\partial r}\dot{g}_{\alpha\beta}
    \sim O(r^{-4}),
\end{gather*}
and find that the contribution of the first term of the integral to the charge is
    \[ 
        V\big(\divergence_b\dot{g}-d\trace_b\dot{g}\big)(\partial_r^{\mbox{\footnotesize{tang}}})=V\left[\frac{1+r^2}{r}\dot{g}_{rr}+\frac{2}{r^3}\sigma^{\alpha\beta}\dot{g}_{\alpha\beta} \right].
    \]
    Computing the remaining terms:
    \begin{gather*}
    -(\trace_b\dot{g})b(\nabla V,\partial_r^{\mbox{\footnotesize{tang}}})
    +2(\trace_b\dot{K})b(\nabla V,\partial_r^{\mbox{\footnotesize{tang}}}) = 
    \frac{r}{\sqrt{1+r^2}}\trace_b(\dot{g}+2p), \\
    \dot{g}(\nabla V,\partial_r^{\mbox{\footnotesize{tang}}}) 
    -2\dot{K}(\nabla V,\partial_r^{\mbox{\footnotesize{tang}}})= r\sqrt{1+r^2}\left[-\dot{g}_{rr}-2p_{rr}\right],
\end{gather*}
 and combining all elements, we obtain a simplified formula for the energy of asymptotically hyperboloidal initial data
    \begin{equation}
        E=\frac{1}{16\pi}\lim_{r\rightarrow\infty} \int_{\Sph^{2}_{r}} \left[\frac{2V^4}{r}\dot{g}_{rr}+\frac{3r^2+2}{r^3}\trace_{\sigma}(\dot{g}_{\alpha\beta})+\frac{2}{r}\trace_{\sigma}(p_{\alpha\beta}) \right]dA_{\Sph^2_r}.
    \end{equation}
    
    We conclude that the energy of an asymptotically hyperboloidal initial data set in the sense of Definition \ref{Def::AHyperboloidal} is well defined as the charge \'a la Michel corresponding to the time translation symmetry initial data on $(\Hyp^3,b,b)$.
    In particular, using the fact that the $O(1)$-elements in the expansion of $\dot{g}_{\alpha\beta}$ and of $p_{\alpha\beta}$ are traceless with respect to the metric on $\Sph^2$, we obtain the expression for energy \begin{equation*}
        E = \frac{1}{16\pi}\int_{\Sph^{2}} \left(2\mathbf{m}_{rr} + 3 \, \trace_{\sigma}\ \!^{g}\mathbf{m}+ 2\trace_{\sigma}\ \!^{k}\mathbf{m}\right) \: dA_{\Sph^{2}}.
    \end{equation*}
    \\
    \emph{Linear momentum:} 
    The KIDs on $(\Hyp^3,b,b)$ corresponding to the spatial translation symmetries in Minkowski are
    \begin{equation*}
        \begin{matrix}
            \mathcal{V}_1= \{V_1= -r\sin\theta\cos\varphi, & Y_1= -\grad_bV_1\},
            \\ 
            \mathcal{V}_2=\{V_2= -r\sin\theta\sin\varphi, &Y_2=-\grad_bV_2 \},
            \\
            \mathcal{V}_3=\{V_3= -r\cos\theta,\ \ \ \ \ \ \  & Y_3= -\grad_bV_3\}.
        \end{matrix}
    \end{equation*}
    Substituting $\mathcal{V}_i$ in \eqref{Eq::MichelCharge} and following the same computation as for energy, we obtain 
    \begin{equation*}
         P^{i} = \frac{1}{16\pi}\int_{\Sph^{2}}\left( 2\mathbf{m}_{rr} + 3 \, \trace_{\sigma}^{g} \mathbf{m} + 2\trace_{\sigma}^{k} \mathbf{m} \right) x^{i} \, dA_{\Sph^{2}}, \end{equation*}
    where $x^{i},\ i=1,2,3$ denote the first spherical harmonics on the unit sphere $\Sph^2$. 
    \end{proof}
    
    \begin{corollary}
        In the following, we will use the following simplified formula for the energy of an asymptotically hyperboloidal initial data set, that follows from the computations done in the proof of Theorem \ref{THM::Energy_and_Linearmomentum_AH}:
        \begin{equation}\label{eq::simplified_energyAH}
             E= \frac{1}{16\pi}\lim_{r\rightarrow\infty} \int_{\Sph^{2}_{r}} V^2\big(\divergence_{b} \:\dot{g}-d \: \trace_{b}\dot{g}\big) -V\trace_{\sigma}(\dot{g}_{\alpha\beta})+ 2V\trace_{\sigma}(\dot{K}_{\alpha\beta}) \:dA_{\Sph^2_r},
        \end{equation}
        and
        \begin{equation}\label{eq::simplified_momentumAH}
            P^{i}= \frac{1}{16\pi}\lim_{r\rightarrow\infty} \int_{\Sph^{2}_{r}} \left(V^2\big(\divergence_{b} \:\dot{g}-d \: \trace_{b}\dot{g}\big) -V\trace_{\sigma}(\dot{g}_{\alpha\beta})+ 2V\trace_{\sigma}(\dot{K}_{\alpha\beta})\right) \: x^{i} \:dA_{\Sph^2_r}.
        \end{equation}
    \end{corollary}

    \begin{example}[Schwarzschild]
        As seen in Example \ref{example:SchwarzschildAH}, we denote by $(\Hyp^3_{S},g_S,g_s)$ the hyperboloid in Schwarzschild spacetime of mass $m>0$. In standard Schwarzschildean coordinates,  $(\Hyp^3_{S},g_S,g_s)$ is asymptotically hyperboloidal in the sense of Definition \ref{Def::AHyperboloidal}. We can then compute the energy of such initial data using Theorem \ref{THM::Energy_and_Linearmomentum_AH} and obtain
        \begin{align*}
            E=&\frac{1}{16\pi}\int_{\Sph^{2}} 2\mathbf{m}_{rr} + 3 \, \trace_{\sigma}\ \!^{g}\mathbf{m}+ 2\trace_{\sigma}\ \!^{k}\mathbf{m} \: dA_{\Sph^{2}}\\
            =& \frac{1}{16\pi}\int_{\Sph^{2}} 4m\ dA_{\Sph^{2}} = m.
        \end{align*}
    \end{example}

\section{Choice of Evolution}\label{sec4}
The Einstein evolution equations, as written in \eqref{eq::Vacuum-EvoEqns}, describe the evolution of the metric and second fundamental form in the normal direction to the initial data set in consideration. This direction for the evolution is admissible, in principle, as the unit normal vector field is everywhere timelike and leads to a definition of a (local) time function. The goal of this section is to clarify why the choice of a time function describing the evolution must be made carefully, to ensure that the foliation it produces satisfies specific properties, which are essential for proving the theorems in Section \ref{sec5}. The key requirement is that the foliation of level sets of the time function at infinity is hyperboloidal. In the Penrose picture, this produces a natural foliation of $\mathscr{I}^{+}$, which is a necessary condition to enable the measurement of energy loss due to radiation. This fact is well known: for a clear and complete review of this, see e.g. \cite{Anil}. 
Here, we review the procedure of finding a suitable $\tau$ which produces a (local) hyperboloidal foliation in exact Minkowski spacetime as done in \cite{Anil}.

We start from the standard Minkowski polar coordinates $(t,r,\theta,\phi)$. Following \cite{Beig-Murchadha}, we consider \begin{equation} 
    \tau := t - h(r).
\end{equation} 

The \emph{height function} $h(r)$ needs to be chosen appropriately to ensure that $\nabla\tau$ is timelike near infinity. Consider the $(3+1)$-decomposition of Minkowski spacetime $(\mathbb{R}^{3,1}, \eta)$ in terms of the temporal function $\tau$:
\begin{equation}\label{metric_tau}
ds^{2} = -d\tau^{2} - 2H(r) d\tau dr + (1 - H(r)^{2}) dr^{2} + r^{2} d\sigma^{2},
\end{equation}
where $H(r) := h'(r)$. Requiring that $\{\tau = \text{constant}\}$ hypersurfaces are radius 1-hyperboloids imposes the condition
\begin{equation}\label{choice_H}
H(r) = 1 - \frac{1}{2r^{2}} + O_2(r^{-\lambda}),
\end{equation}
with $\lambda\geq3.$ 
For $H(r)$ as in \eqref{choice_H}, the time function takes the form $\tau = t - r - \frac{1}{r} + O_3\left(r^{-2}\right)$. Thus, asymptotically, $\tau$ approaches the outgoing null coordinate $u := t - r$, leading to a foliation of $\mathscr{I}^{+}$.
\begin{figure}
    \centering
    \includegraphics[width=0.4\linewidth]{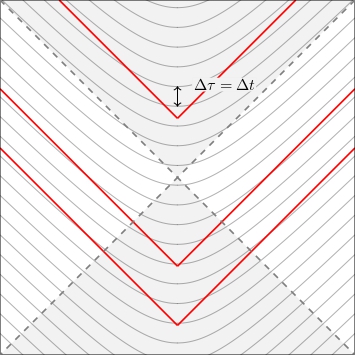} \includegraphics[width=0.4\linewidth]{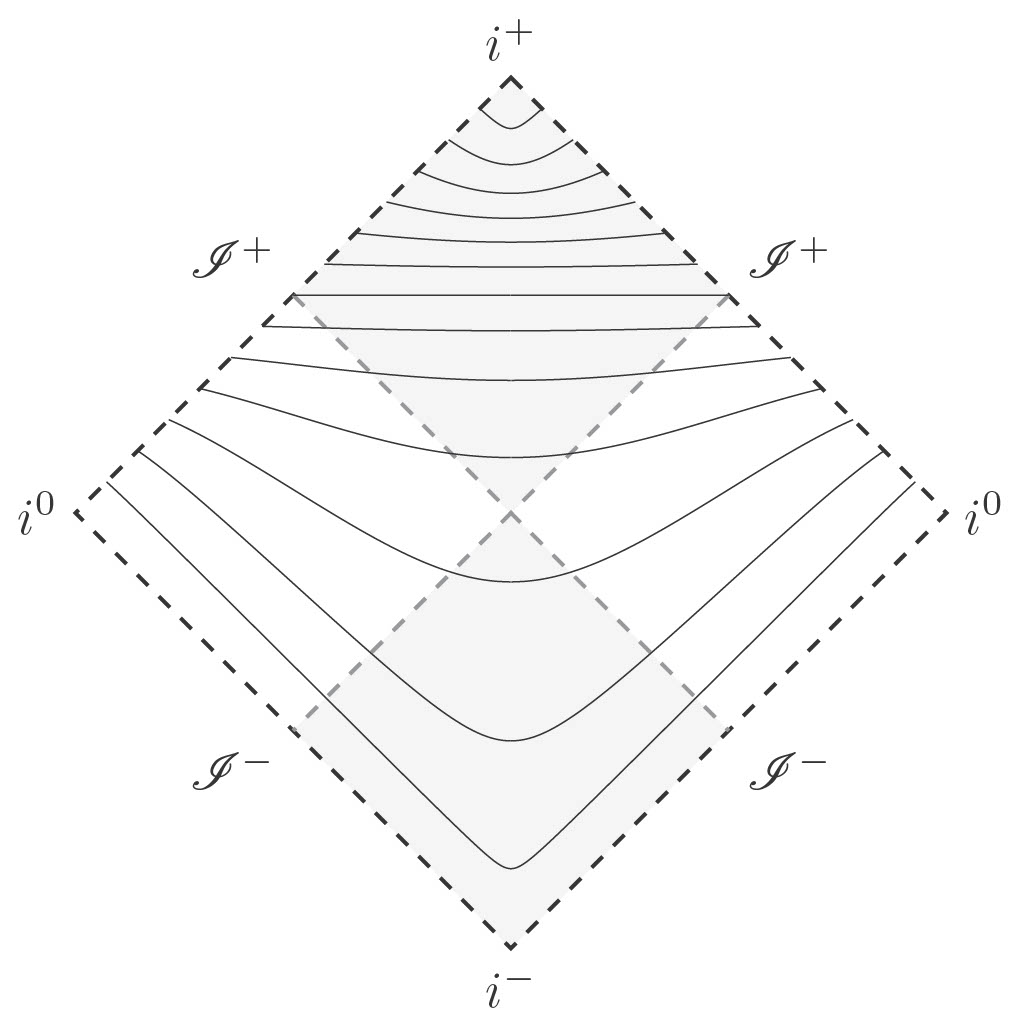}
    \caption{Foliation of Minkowski spacetime obtained by evolving the standard hyperboloid $(\Hyp^3,b,b)$ with the standard Killing time $\partial_\tau=\partial_t$. The red lines represent radiation, hitting only some of the future hyperboloids, as can be seen in the first picture. The second picture shows the foliation of $\mathscr{I}^+$ produced by the hyperboloidal foliation. 
    There are many functions $h(r)$ that satisfy \eqref{choice_H}, and we refer the interested reader to \cite{AnilZenginolu2025} for some explicit examples of height functions. See \cite{AnilDiagrams} for a guide to drawing Penrose diagrams.}
    \label{fig:enter-label}
\end{figure}

\begin{remark}\label{rmk:standardhyp}
    The standard hyperboloidal initial data $\left(\Hyp^{3},b,b\right)$ is obtained by considering the generalized hyperboloidal initial data with 
    $$h(r) = \sqrt{1+r^{2}},$$ and the lapse and shift on this initial data corresponding to the standard Killing time in Minkowski are
\begin{equation}\label{Eqn::standardLapseandShift}
    N=\sqrt{1+r^2}\ , \qquad X^r= -r\sqrt{1+r^{2}}.
\end{equation}
\end{remark}

It can be easily checked that the standard hyperboloidal data is invariant under the evolution in the Killing time direction. In particular,
\begin{equation*}
    \mathcal{L}_{\partial_\tau}{b} = \Lie_{N\nu+X} =0,
\end{equation*}
where $\nu$ is the unit normal to $(\Hyp^3,{b},b)$ in Minkowski, $N$ and $X$ are computed as in \eqref{Eqn::standardLapseandShift}, see Appendix \ref{App:evolution}. This is not coincidental; rather, it results from the choice of evolution being dictated by the exact Killing initial data (KID) of the initial data set. This property is crucial for establishing a well-defined notion of charges along the foliation defined by $\tau$, but it is not sufficient. In addition, we need the charge integrals to converge on each leaf of the foliation. If the evolution of the initial data $\idsmatter$ satisfies this last condition, we can measure the evolution of energy and linear momentum along the foliation. This is dealt with in detail in the next section, see the discussion following Definition \ref{Def:E-P_chargeability}.

\section{Evolution of Michel Charges on Asymptotically Hyperboloidal Foliations}\label{sec5}

In this section, we examine the evolution of energy and linear momentum for asymptotically hyperboloidal initial data sets which are vacuum outside a compact set. As previously discussed, we employ an evolution that preserves the background structure. The following theorems describe how the energy and linear momentum evolve along the foliation introduced earlier. This evolution is linked to the outgoing gravitational radiation measured on each foliation slice, resulting in a monotonically decreasing energy and linear momentum. We first present the evolution of the energy of asymptotically hyperboloidal initial data sets.

\begin{theorem}\label{thm:evolutionE}
Let $\idsmatter$ be an asymptotically hyperboloidal initial data set with respect to a diffeomorphism  
$
\Psi: \Hyp^{3} \setminus B_0 \to M \setminus B
$, where $B\subset M, B_0\subset\Hyp^3$ are compact sets, and assume that $\supp(\rho), \supp(J)\subset B$.
Let \(\tau\) be the standard time function defining a foliation of standard hyperboloids in Minkowski, and let $N$ and $X$ be the corresponding lapse and shift \eqref{Eqn::standardLapseandShift}.
We denote the foliation of the globally hyperbolic development of \((M, g, K)=\{\tau=0\}\), by \((M \times \{\tau\}, g(\tau), K(\tau))\), with \(\tau \in (-\varepsilon, \varepsilon)\), for some $\varepsilon >0$.

Then, the energy \(E(\tau)\) is well-defined on each leaf of the foliation \(\left(M \times \{\tau\},g(\tau),K(\tau)\right)\) and satisfies  
\begin{equation}
    \frac{\partial}{\partial\tau}E(\tau) 
\;=\;
-\frac{1}{4\pi} 
\int_{\Sph^{2}} 
\bigl|{}^{0}\mathbf{g} + {}^{0}\mathbf{p}\bigr|_{\sigma}^{2}
\,dA_{\Sph^{2}},
\end{equation}
where $|{}^{0}\mathbf{g} + {}^{0}\mathbf{p}|^{2}_{\sigma} = \left({}^{0}\mathbf{g}^{\alpha\beta} + {}^{0}\mathbf{p}^{\alpha\beta}\right)\left({}^{0}\mathbf{g}_{\alpha\beta} + {}^{0}\mathbf{p}_{\alpha\beta}\right)$.
\end{theorem}

\begin{proof}

   We start by recalling that the energy is well defined on the $\tau=0$ slice using Definition \ref{Def::AHyperboloidal}, as seen in Theorem \ref{THM::Energy_and_Linearmomentum_AH}. The Einstein evolution equations in vacuum are given by
    \begin{equation}
        \frac{\partial}{\partial\tau}g = 2NK + \mathcal{L}_{X}g,
    \end{equation}
    and
    \begin{equation}
        \frac{\partial}{\partial\tau}K = \Hess N - N\Ric + 2N(K \circ K) - NHK + \mathcal{L}_{X}K,
    \end{equation}
where $(K \circ K) = K_{i}^{k}K_{kl}$ and $H = \trace_gK$. Using the above in the expression for energy \eqref{eq::simplified_energyAH}, we obtain
\begin{equation}
    \begin{aligned}\label{evolvedcharge}
        \frac{\partial}{\partial \tau} E(\tau) 
        &= \frac{1}{16\pi} \lim_{r \rightarrow \infty}  
        \int_{\mathbb{S}^{2}_{r}} \Bigg[  
            V^{2} 
            \divergence_{b} (2NK + \mathcal{L}_{X} g)(\partial_{r})  \\
        &\quad - V^{2} d \trace_{b} (2NK + \mathcal{L}_{X} g)(\partial_{r})  
            - V \trace_{b} (2NK_{\alpha\beta} + \mathcal{L}_{X} g_{\alpha\beta})  \\
        &\quad + 2 V
            \trace_{b} \Big( \Hess_{\alpha\beta} N - N \Ric_{\alpha\beta}  
            + 2N (K \circ K)_{\alpha\beta} - N H K_{\alpha\beta}  \Big)  \\
        &\quad + 2 V 
            \trace_{b} \Big( \mathcal{L}_{X} K_{\alpha\beta} \Big)  
        \Bigg] \, dA_{\mathbb{S}^{2}_{r}},
    \end{aligned}
\end{equation}
where we have used that $\dot{g} = g - b$ and $\dot{K} = K - b$, and the fact that the background data is invariant under evolution. First, we compute the Lie derivatives of the components of the correction $\dot{g}$ with the shift $X$:
\begin{equation}
    \begin{aligned}
        \mathcal{L}_{X}\dot{g}_{rr} 
        &=-\mathcal{L}_{r^{2}\partial_{r}}\dot{g}_{rr} + O_{1}(r^{-6}) \\
        &= - r^{2}\partial_{r}\dot{g}_{rr} - 2\partial_{r}(r^{2})\dot{g}_{rr} + O_{1}(r^{-6}) \\
        &= \frac{\mathbf{m}_{rr}}{r^{4}} + 2\frac{\tilde{\mathbf{m}}_{rr}}{r^{5}} + O_{1}\left(r^{-6}\right),
    \end{aligned}
\end{equation}
similarly
\begin{equation}
\begin{aligned}
        \mathcal{L}_{X}\dot{g}_{r\alpha} &= -r^{2}\partial_{r}\dot{g}_{r\alpha} - \partial_{r}(r^{2})\dot{g}_{r\alpha} - \partial_{\alpha}(r^{2})\dot{g}_{rr} + O_{1}\left(r^{-4}\right)
        \\&= \frac{\mathbf{m}_{r\alpha}}{r^{2}} + O_{1}\left(r^{-3}\right),
\end{aligned}
\end{equation}
and 
\begin{equation}
\begin{aligned}
        \mathcal{L}_{X}\dot{g}_{\alpha\beta} &= -r^{2}\partial_{r}\dot{g}_{\alpha\beta} - \partial_{\alpha}(r^{2})\dot{g}_{r\beta} - \partial_{\beta}(r^{2})\dot{g}_{r\alpha} + O_{1}\left(r^{-3}\right) 
    \\&= {}^{g}\mathbf{m}_{\alpha\beta} + 2\frac{{}^{g}\tilde{\mathbf{m}}_{\alpha\beta}}{r} + O_{1}\left(r^{-2}\right).
\end{aligned}
\end{equation}

We consider the first term, $\divergence_{b}\left(2NK + \mathcal{L}_{X}g\right)$. Note that this is a one-form acting on $\partial_r$, and hence it suffices to only compute the $r$ component of this term. Using the expression for the divergence of a two-tensor $T$,
\begin{equation*}
    (\divergence_{b}T )_{k} = b^{ij}\left(\partial_{j}T_{ik} - {}^{b}\Gamma_{ij}^{l}T_{kl} - {}^{b}\Gamma^{l}_{jk}T_{il}\right),
\end{equation*}
we compute the following:
\begin{equation}
\begin{aligned}
    \divergence_{b}\left(2NK + \mathcal{L}_{X}g\right)_{r} &= -\frac{2}{r^{3}}\trace_{\sigma}\left({}^{g}\mathbf{m} +{}^{k}\mathbf{m}\right) - \frac{1}{r^{3}}\trace_{\sigma}{}^{g}\mathbf{m} - \frac{2}{r^{4}} \left(\tilde{\mathbf{m}}_{rr} + \tilde{\mathbf{p}}_{rr}\right) \\&  - \frac{2}{r^{4}}{}^{g}\tilde{\mathbf{m}}_{rr}  + \frac{2}{r^{4}}\divergence_{\sigma}\left(\mathbf{g}_{r\alpha} + \mathbf{p}_{r\alpha}\right) + \frac{1}{r^{4}}\divergence_{\sigma}\mathbf{g}_{r\alpha}  \\& - \frac{2}{r^{4}}\trace_{\sigma}\left({}^{g}\tilde{\mathbf{m}} +{}^{k}\tilde{\mathbf{m}}\right)   -\frac{2}{r^{4}}\trace_{\sigma}{}^{g}\tilde{\mathbf{m}} + O\left(r^{-5}\right).
\end{aligned}
\end{equation}

The leading-order terms in this expression decay as \(O(r^{-3})\). When multiplied by the factor of order \(r^2\) in Equation (\ref{evolvedcharge}), these terms would generically yield a divergent contribution in the limit of integration over a sphere of radius \(r\). In contrast, the subleading terms, which decay as \(O(r^{-4})\), remain integrable in this limit. This behavior is common to other terms in the evolution of the energy.

We will demonstrate that the leading-order divergences arising from the evolution of the metric and the extrinsic curvature cancel exactly, ensuring that the resulting flux integral remains finite.

The second term is computed similarly,
\begin{equation}
    \begin{aligned}
        \differential\trace_{b}\left(2NK + \mathcal{L}_{X}g\right) &= -\frac{4}{r^{3}}\left(\mathbf{m}_{rr} + \mathbf{p}_{rr}\right) - \frac{2}{r^{3}} \mathbf{m}_{rr} - \frac{4}{r^{3}}\trace_{\sigma}\left({}^{g}\mathbf{m} + {}^{k}\mathbf{m}\right) \\&  - \frac{2}{r^{3}}\trace_{\sigma}{}^{g}\mathbf{m} - \frac{6}{r^{4}}\left(\tilde{\mathbf{m}}_{rr} + \tilde{\mathbf{p}}_{rr}\right) - \frac{6}{r^{4}}\tilde{\mathbf{m}}_{rr} \\& - \frac{6}{r^{4}}\trace_{\sigma}\left({}^{g}\tilde{\mathbf{m}} + {}^{k}\tilde{\mathbf{m}}\right)- \frac{6}{r^{4}}\trace_{\sigma}\tilde{{}^{g}\mathbf{m}} + O\left(r^{-5}\right).
  \end{aligned}
\end{equation}

The following term is given by
\begin{equation}
    \begin{aligned}
        \frac{1}{r^{2}}\trace_{\sigma} (2NK + \mathcal{L}_{X} g) &= \frac{2}{r^{2}}\trace_{\sigma}\left({}^{g}\mathbf{m} +{}^{k}\mathbf{m}\right) + \frac{1}{r^{2}}\trace_{\sigma}{}^{g}\mathbf{m} 
        \\& + \frac{2}{r^{3}}\trace_{\sigma}\left({}^{g}\tilde{\mathbf{m}} + {}^{k}\tilde{\mathbf{m}}\right) + \frac{2}{r^{3}} \trace_{\sigma}{}^{g}\tilde{\mathbf{m}} \\& + O_{1}\left(r^{-4}\right).
    \end{aligned}
\end{equation}

The next set of terms are obtained from the evolution of the angular components of the extrinsic curvature $K_{\alpha\beta}$. The first term is the angular components of the $g$-Hessian of the lapse function. Since $N$ does not depend on the angular coordinates,
\begin{equation}
    \begin{aligned}
        \frac{1}{r^{2}}\sigma^{\alpha\beta}\Hess N_{\alpha\beta} &= -\frac{1}{r^{2}}\sigma^{\alpha\beta}\Gamma^{r}_{\alpha\beta} \\& = 
         -\frac{2}{r^{2}}\mathbf{m}_{rr} - \frac{1}{2r^{2}}\trace_{\sigma}{}^{g}\mathbf{m} - \frac{2}{r^{3}}\tilde{\mathbf{m}}_{rr} 
        \\& -\frac{1}{r^{3}}\divergence_{\sigma}\mathbf{g}_{r\alpha} - \frac{1}{r^{3}}\trace_{\sigma}{}^{g}\tilde{\mathbf{m}}  + O\left(r^{-4}\right).
    \end{aligned}
\end{equation}

The next term is the $b$-trace of $N\Ric_{\alpha\beta}$. We have:
\begin{equation*}
\begin{aligned}
    \Ric_{\alpha\beta} &= \Gamma_{\alpha\beta,r}^{r} + \Gamma_{rr}^{r}\Gamma^{r}_{\alpha\beta} + \Gamma_{\delta r}^{\delta}\Gamma^{r}_{\alpha\beta} - \Gamma^{\delta}_{\beta r}\Gamma_{\alpha\delta}^{r} - \Gamma_{\beta\delta}^{r}\Gamma^{\delta}_{\alpha r}
    \\& - \Gamma^{r}_{\alpha r,\beta} + \Gamma^{r}_{r\delta}\Gamma^{\delta}_{\alpha\beta} - \Gamma^{r}_{\beta r}\Gamma^{r}_{\alpha r} + \Gamma^{\delta}_{\alpha\beta,\delta} - \Gamma^{\delta}_{\alpha\delta,\beta} 
    \\& + \Gamma^{\delta}_{\delta\rho}\Gamma^{\rho}_{\alpha\beta} - \Gamma^{\delta}_{\beta\rho}\Gamma^{\rho}_{\alpha\delta}.
\end{aligned}
\end{equation*}

We first note that the last four terms have cancellations at leading order and contribute only at lower order. Using the expansions of the Christoffel symbols in Appendix \ref{appA}, we find the following:
\begin{equation}
    \begin{aligned}
        \frac{1}{r^{2}}\sigma^{\alpha\beta}N\Ric_{\alpha\beta} & = \frac{1}{r^{2}}\mathbf{m}_{rr} - \frac{1}{2r^{2}}\trace_{\sigma}{}^{g}\mathbf{m} +\frac{4}{r^{3}}\tilde{\mathbf{m}}_{rr} - \frac{2}{r^{3}}\trace_{\sigma}{}^{g}\tilde{\mathbf{m}}
        + O\left(r^{-4}\right).
    \end{aligned}
\end{equation}

The following term is given by:
\begin{equation}
    \begin{aligned}
        \frac{2}{r^{2}}N\sigma^{\alpha\beta} (K \circ K)_{\alpha\beta} &= \frac{2}{r^{2}}N\sigma^{\alpha\beta} \left(g^{rr}K_{\alpha r}K_{\beta r} + g^{r\rho}K_{\alpha r}K_{\beta \rho} + g^{\rho \delta}K_{\alpha \rho}K_{\beta \delta}\right)   
        \\& = -\frac{2}{r^{2}} \trace_{\sigma}{}^{g}\mathbf{m} + \frac{4}{r^{2}} \trace_{\sigma}\left({}^{g}\mathbf{m} + {}^{k}\mathbf{m}\right) - \frac{2}{r^{3}}|{}^{0}\mathbf{g} + {}^{0}\mathbf{p}|^{2}_{\sigma}
        \\& - \frac{2}{r^{3}}\trace_{\sigma}{}^{g}\tilde{\mathbf{m}} + \frac{4}{r^{3}}\trace_{\sigma}\left({}^{g}\tilde{\mathbf{m}} + {}^{k}\tilde{\mathbf{m}}\right) + O_{1}\left(r^{-4}\right),
    \end{aligned}
\end{equation}
where $|{}^{0}\mathbf{g} + {}^{0}\mathbf{p}|^{2}_{\sigma}$ is the contribution from radiation. Following this, we have:
\begin{equation}
    \begin{aligned}
        \frac{1}{r^{2}}N\sigma^{\alpha\beta}HK_{\alpha\beta} &= \frac{N}{r^{2}}\sigma^{\alpha\beta}\left(3\dot{K}_{\alpha\beta} + \trace_{b}\dot{K}b_
        {\alpha\beta} + \dot{g}^{ij}b_{ij}b_{\alpha\beta}\right) + O_{1}\left(r^{-5}\right)
        \\& = \frac{5}{r^{2}}\trace_{\sigma}\left({}^{g}\mathbf{m} + {}^{k}\mathbf{m}\right) + \frac{2}{r^{2}}\left(\mathbf{m}_{rr} + \mathbf{p}_{rr}\right)  - \frac{2}{r^{2}}\mathbf{m}_{rr} \\& 
         - \frac{2}{r^{2}}\trace_{\sigma}{}^{g}\mathbf{m}  - \frac{2}{r^{3}}\tilde{\mathbf{m}}_{rr}  +\frac{2}{r^{3}}\left(\tilde{\mathbf{m}}_{rr} + \tilde{\mathbf{p}}_{rr}\right)\\&  - \frac{2}{r^{3}}\trace_{\sigma}{}^{g}\tilde{\mathbf{m}} + \frac{5}{r^{3}}\trace_{\sigma}\left({}^{g}\tilde{\mathbf{m}} + {}^{k}\tilde{\mathbf{m}}\right)+ O_{1}\left(r^{-4}\right).
    \end{aligned}
\end{equation}

The final term is the Lie derivative of the angular components of the extrinsic curvature with respect to the shift $X$,
\begin{equation}
    \begin{aligned}
        \frac{1}{r^{2}}\sigma^{\alpha\beta}\mathcal{L}_{X}K_{\alpha\beta} &=  \frac{1}{r^{2}}\trace_{\sigma}\left({}^{g}\mathbf{m} + {}^{k}\mathbf{m}\right) + \frac{2}{r^{3}}\trace_{\sigma}\left({}^{g}\tilde{\mathbf{m}} + {}^{k}\tilde{\mathbf{m}}\right) + O\left(r^{-4}\right).
    \end{aligned}
\end{equation}

We summarize the evolution of the extrinsic curvature components $K_{\alpha\beta}$ upon tracing with $\sigma$,
\begin{equation}
\begin{aligned}
    \frac{1}{r^{2}}\sigma^{\alpha\beta}\frac{\partial}{\partial\tau}K_{\alpha\beta} &= -\frac{1}{r^{2}}\mathbf{m}_{rr} - \frac{2}{r^{2}}\left(\mathbf{m}_{rr} + \mathbf{p}_{rr}\right) - \frac{2}{r^{3}}|{}^{0}\mathbf{g}+{}^{0}\mathbf{p}|^{2}_{\sigma} \\&  + \frac{1}{r^{3}}\trace_{\sigma}{}^{g}\tilde{\mathbf{m}} + \frac{1}{r^{3}}\trace_{\sigma}\left({}^{g}\tilde{\mathbf{m}} + {}^{k}\tilde{\mathbf{m}}\right) - \frac{2}{r^{3}}\tilde{\mathbf{m}}_{rr}  \\&- \frac{2}{r^{3}}\left(\tilde{\mathbf{m}}_{rr} + \tilde{\mathbf{p}}_{rr}\right) + O\left(r^{-4}\right).
\end{aligned}
\end{equation}

Now, we put together everything in the integral in equation (\ref{evolvedcharge}). After several cancellations, we obtain
\begin{align}
\frac{\partial}{\partial\tau}E(\tau) &= \frac{1}{16\pi} \lim_{r \rightarrow \infty} \int_{\Sph^{2}_{r}}  
\Bigg( - \frac{4}{r^{2}} |{}^{0}\mathbf{g} + {}^{0}\mathbf{p}|_{\sigma}^{2}  
+ \frac{2}{r^{2}} \divergence_{\sigma} \mathbf{p}_{r\alpha}  
+ \frac{1}{r^{2}} \divergence_{\sigma} \mathbf{g}_{r\alpha}  \Bigg) \, dA_{\Sph^{2}_{r}}  \notag \\ 
&\quad + \frac{1}{16\pi} \lim_{r \rightarrow \infty} \int_{\Sph^{2}_{r}}  
\Big( \frac{4}{r^{2}}\trace_{\sigma} {}^{g} \tilde{\mathbf{m}}  
+ \frac{4}{r^{2}} \trace_{\sigma} \left({}^{g} \tilde{\mathbf{m}} + {}^{k} \tilde{\mathbf{m}}\right) \Big) \, dA_{\Sph^{2}_{r}},
\end{align}
where terms of the decay rate $r^{-3}$ have been discared as they vanish under the limit $r \rightarrow\infty$.

The divergence theorem on closed manifolds ensures that the terms that are total divergences on the sphere, i.e., $\divergence_{\sigma} \mathbf{p}_{r\alpha}$ and $\divergence_{\sigma} \mathbf{g}_{r\alpha}$ integrate to zero. Moreover, since our asusmptions ask for $\tilde{{}^{g}\mathbf{m}}_{\alpha\beta}$ and $\tilde{{}^{k}\mathbf{m}}_{\alpha\beta}$ to be tracefree, the traces of these quantites vanish in the above integral. Therefore, we obtain
\begin{equation*}
    \frac{\partial}{\partial\tau}E(\tau)= -\frac{1}{4\pi} \int_{\Sph^{2}}|{}^{0}\mathbf{g} + {}^{0}\mathbf{p}|^{2}_{\sigma} \: dA_{\Sph^{2}}.
\end{equation*}
\end{proof}

Having established the evolution of energy, we now turn to the evolution of linear momentum. The following theorem governs the loss of linear momentum and requires additional assumptions on the components of \( \dot{g}_{r\alpha} \) and \( \dot{p}_{r\alpha} \). Since much of the proof of Theorem \ref{thm:evolution_P} follows the same steps as Theorem \ref{thm:evolutionE}, we highlight only the key differences.
\begin{theorem}\label{thm:evolution_P}
Let $\idsmatter$ be an asymptotically hyperboloidal initial data set with respect to a diffeomorphism  
$
\Psi: \Hyp^{3} \setminus B_0 \to M \setminus B
$, where $B\subset M, B_0\subset\Hyp^3$ are compact sets, and assume that $\supp(\rho), \supp(J)\subset B$. Let $\tau$ be the time function defining a foliation of standard hyperboloids in Minkowski, and let $N$ and $X$ be the corresponding lapse and shift (\ref{Eqn::standardLapseandShift}). 

We denote the foliation of the globally hyperbolic development of $\left(M,g,K\right)$ by $\left(M\times\{\tau\},g(\tau),K(\tau)\right)$, with $\tau \in (-\varepsilon,\varepsilon)$ for some $\varepsilon > 0$.

Assume that the one-forms $\mathbf{g}_{r\alpha}$ and $\mathbf{p}_{r\alpha}$ are $\sigma$-divergences of symmetric traceless 2-tensors, i.e.,  
\[
\mathbf{g}_{r\alpha} = \divergence_{\sigma}{}^{g}\mathbf{T}_{\alpha\beta}, \quad  
\mathbf{p}_{r\alpha} = \divergence_{\sigma}{}^{k}\mathbf{T}_{\alpha\beta},
\]
where ${}^{g}\mathbf{T}_{\alpha\beta}$ and ${}^{k}\mathbf{T}_{\alpha\beta}$ are trace-free symmetric 2-tensors on the unit sphere.

Then the linear momentum $P^i(\tau)$ is well-defined on each leaf of the foliation $\left(M\times\{\tau\},g(\tau),K(\tau)\right)$ and satisfies  
\begin{equation}\label{equation_evo_P}
\frac{\partial}{\partial\tau}P^i(\tau)
= -\frac{1}{4\pi}
\int_{\Sph^{2}} 
\bigl|{}^{0}\mathbf{g} + {}^{0}\mathbf{p}\bigr|_{\sigma}^{2}
\,x^{i}\,dA_{\Sph^{2}}.
\end{equation}
\end{theorem}

\begin{proof}
    The proof follows analogously to the proof of Theorem \ref{thm:evolutionE}, with the only difference arising from the divergence terms, which do not integrate to zero a priori. We obtain  
    \begin{equation}\label{eqn:linearmomentumextraterms}
        \frac{\partial}{\partial\tau}P^{i} = \frac{1}{16\pi}\lim_{r\rightarrow\infty}\int_{\Sph^{2}_{r}}\left(-\frac{4}{r^{2}}|{}^{0}\mathbf{g} + {}^{0}\mathbf{p}|^{2}_{\sigma} + \frac{2}{r^{2}}\divergence_{\sigma}\mathbf{p}_{r\alpha} + \frac{1}{r^{2}}\divergence_{\sigma}\mathbf{g}_{r\alpha}\right)x^{i}dA_{\Sph^{2}_{r}}.
    \end{equation}

    We now focus on the divergence terms in the integral. We rewrite the above integral into a total divergence and rest terms,
    \begin{equation*}
    \begin{aligned}    
        \int_{\Sph^{2}}\nabla^{\alpha}\left(2\nabla^{\beta}{}^{k}\mathbf{T}_{\alpha\beta} +  \nabla^{\beta}{}^{g}\mathbf{T}_{\alpha\beta}\right)x^{i} dA_{\Sph^{2}} &= \int_{\Sph^{2}}\nabla^{\alpha}\left(2\nabla^{\beta}{}^{k}\mathbf{T}_{\alpha\beta}x^{i} + \nabla^{\beta}{}^{g}\mathbf{T}_{\alpha\beta}x^{i}\right)dA_{\Sph^{2}} \\ 
        & \quad -  \int_{\Sph^{2}} \left(2\nabla^{\alpha}{}^{k}\mathbf{T}_{\alpha\beta} + \nabla^{\alpha}{}^{g}\mathbf{T}_{\alpha\beta}\right)\nabla^{\beta}x^{i}dA_{\Sph^{2}},
    \end{aligned}
    \end{equation*}
    where $\nabla$ is the Levi-Civita connection on the unit sphere.  

    The first term on the right-hand side vanishes due to the divergence theorem on closed manifolds. Splitting up the second term as before, we obtain
    \begin{equation*}
        \begin{aligned}
            \int_{\Sph^{2}} \left(2\nabla^{\alpha}{}^{k}\mathbf{T}_{\alpha\beta} + \nabla^{\alpha}{}^{g}\mathbf{T}_{\alpha\beta}\right)\nabla^{\beta}x^{i} dA_{\Sph^{2}} &= \int_{\Sph^{2}}\nabla^{\alpha}\left(\left(2{}^{k}\mathbf{T}_{\alpha\beta} + {}^{g}\mathbf{T}_{\alpha\beta}\right)\nabla^{\beta}x^{i}\right)dA_{\Sph^{2}}
            \\& \quad -\int_{\Sph^{2}}\left(2{}^{k}\mathbf{T}_{\alpha\beta} + {}^{g}\mathbf{T}_{\alpha\beta}\right)\nabla^{\alpha}\nabla^{\beta}x^{i} dA_{\Sph^{2}}.
        \end{aligned}
    \end{equation*}

    Again, the first term on the right-hand side vanishes by the divergence theorem. For the second term, recall that the $x^{i}$ are the first spherical harmonics on the unit sphere, so  
    \[
    \nabla^{\alpha}\nabla^{\beta}x^{i} = -x^{i}\sigma^{\alpha\beta}.
    \]  
    Since ${}^{k}\mathbf{T}_{\alpha\beta}$ and ${}^{g}\mathbf{T}_{\alpha\beta}$ are trace-free, this term also vanishes.  

    Thus, we obtain the evolution equation for $P^{i}$ as given in (\ref{equation_evo_P}).  
\end{proof}

As seen in the above proofs, the asymptotic structure of an asymptotically hyperboloidal initial data set from Definition \ref{Def::AHyperboloidal} is not preserved under the Einstein evolution equations. In particular, starting from an asymptotically hyperboloidal initial data $\idsmatter$, at $\tau = 0$, and evolving using the standard lapse $N$ and shift $X$ as in Remark \ref{rmk:standardhyp}, we obtain that the asymptotic behavior of the geometry on the $\tau = 1$ leaf $\left(M,\bar{g},\bar{K}\right)$ is as follows:
    \begin{gather}
        \bar{g}_{rr} = g_{rr} + \frac{\mathfrak{f}(u^{A})}{r^{4}},  \\
        \bar{g}_{\alpha\beta} = g_{\alpha\beta} + 2r\left({}^{0}\mathbf{g}_{\alpha\beta} + {}^{0}\mathbf{p}_{\alpha\beta}\right) + \mathfrak{a}_{\alpha\beta} + \frac{1}{r}\mathfrak{b}_{\alpha\beta}, \\
        \bar{K}_{\alpha\beta} = K_{\alpha\beta} - \mathfrak{f}(u^{A})\sigma_{\alpha\beta} -\frac{2}{r} |{}^{0}\mathbf{g}+ {}^{0}\mathbf{p}|^{2}_{\sigma}\sigma_{\alpha\beta}.
    \end{gather}
    
    Here, $\mathfrak{f} \in C^{2}(\Sph^{2})$, while $\mathfrak{a}_{\alpha\beta} \in C^{2}(\Sph^{2})$, and $\mathfrak{b}_{\alpha\beta} \in C^{1}(\Sph^{2})$, with $\mathfrak{b}_{\alpha\beta}$ trace-free. The terms $g_{rr}$, $g_{\alpha\beta}$, and $K_{\alpha\beta}$ are as in Definition \ref{Def::AHyperboloidal}. The remaining components $\bar{g}_{r\alpha}$, $\bar{K}_{rr}$, and $\bar{K}_{r\alpha}$ evolve according to the Einstein evolution equations, but they are not written explicitly here since  they do not contribute to the computation of the charges $E$ and $P^{i}$ at the orders considered.

    Then, $\left(M,\bar{g},\bar{K}\right)$ is not asymptotically hyperboloidal in the sense of Definition \ref{Def::AHyperboloidal}. However, due to delicate cancellations between the components of $\bar{g}_{rr}$ and $\bar{K}_{\alpha\beta}$ the energy and linear momentum charges are still well-defined in the sense of Michel. 
    
     Since the constraint equations propagate under the Einstein evolution equations \cite{Bartnik2004}, the above data set satisfies the constraints. The energy of the evolved initial data $\left(M,\bar{g},\bar{K}\right)$ is then given by
    \begin{equation}
        E(\tau = 1) = \frac{1}{16\pi}\lim_{r \rightarrow\infty}\int_{\Sph_{r}^{2}} \left( 2\mathbf{m}_{rr} + 3\trace_{\sigma}{}^{g}\mathbf{m} + 2\trace_{\sigma}{}^{k}\mathbf{m} - 4|{}^{0}\mathbf{g} + {}^{0}\mathbf{p}|_{\sigma}^{2} \right) dA_{\Sph^{2}_{r}}.
    \end{equation}

This implies that there are families of asymptotics other than the ones defined in Definition \ref{Def::AHyperboloidal} that admit well-defined notions of energy and linear momentum. With this in mind, we introduce the following definition of \emph{E-P chargeability \`a la Michel}.

\begin{definition}{(E-P Chargeability)}\label{Def:E-P_chargeability}
    $\idsmatter$ is \emph{E-P chargeable} if there exist a spacelike hypersurface in Minkowski $(M_0,g_0,K_0)$, compact sets $\mathcal{K}_0\subset M_0, \mathcal{K}\subset M$, and a diffeomorphism at infinity 
    \[
        \Psi: M_0\setminus \mathcal{K}_0\rightarrow M\setminus\mathcal{K}
    \]
    such that the following integrals converge:
    \begin{align*}
         m(\dot{g},\dot{K},\mathcal{V})=\lim_{r\rightarrow\infty}\frac{1}{16\pi}\oint_{\Sph_r^2} 
            \big[ 
                V\big(&\divergence_0\dot{g}-d\trace_0\dot{g}\big)
                -\iota_{\nabla V}\dot{g} +(\trace_0\dot{g})dV
                +2\big(\iota_\alpha\dot{K}-(\trace_0\dot{K})\alpha)\big)\\
                &+(\trace_0\dot{g}) \iota_\alpha K_0
                +\langle K_0,\dot{g}\rangle_0\alpha-2\iota_\alpha(\dot{g}\circ K_0) \big](n)
                dA_{\Sph^2_r},
    \end{align*}
    for any $\mathcal{V}\in\{\mathcal{V}_0,\mathcal{V}_1,\mathcal{V}_2,\mathcal{V}_3\}$, where $\mathcal{V}_0, \mathcal{V}_i, i=1,2,3$ are the KIDs on $(M_0,g_0,K_0)$ corresponding to the time and spatial translations in Minkowski.
    Here, $dA_{\Sph^2_r}$ is the area element on the sphere of radius $r$ in $(M_0,g_0)$, and $n$ is the outward unit normal to $\Sph^2_r$ within the background manifold. As before, $\dot{g}:=\Psi^*g-g_0$ and $\dot{K}:=\Psi^*K-K_0$.
\end{definition}

 Therefore, it is clear that $(M, \bar{g}, \bar{K})$ is E-P chargeable. 

 One could apply the Einstein evolution equations to $(M,\bar{g},\bar{K})$. This would again introduce terms that appear to diverge at first glance however, cancellations similar to those observed above ensure that E-P chargeability is preserved by the evolution if natural conditions (similar to Definition \ref{Def::AHyperboloidal} and Theorem \ref{thm:evolution_P}) are imposed on the lower order expansions in powers of $\frac{1}{r}$ of $\dot{g}_{\alpha\beta}$, $p_{\alpha\beta}$, $\dot{g}_{r\alpha}$ and $p_{r\alpha}$, on the original data $(M,g,K)$. 
 
 We observe that the fluxes for energy and linear momentum in our hyperboloidal approach closely mirror those in the Bondi--Sachs (null) formalism. For a review of the Bondi-Sachs formalism and mass loss in the null setting, we refer the reader to \cite{Madler:2016xju}. In particular, the combination ${}^{0}\mathbf{g}_{\alpha\beta} + {}^{0}\mathbf{p}_{\alpha\beta}$ plays a role analogous to the news tensor $N_{\alpha\beta}$, which governs energy loss at null infinity \cite{Bondi-vanderBurg-Metzner,Godazgar2019, RKSachs1962}. The news tensor $N_{\alpha\beta}$ is the retarded-time derivative of the leading-order perturbation to the unit-sphere metric in the Bondi expansion. Moreover, the additional conditions on $\mathbf{g}_{r\alpha}$ and $\mathbf{p}_{r\alpha}$ resemble the conditions that ensure the vanishing of the additional terms in the analog of \eqref{eqn:linearmomentumextraterms} in the null setting. Additionally, as outlined above, the terms in the expansion of components of $\dot{g}$ and $p$ can be recursively determined using the evolution equations. This is also similar to the null picture \cite{Chrusciel:1993hx, Godazgar2019}. This leads to a consistent picture of gravitational energy--momentum flux in the hyperboloidal setting.

  \begin{remark}\label{Anderssoncomparison}
It is shown in \cite{Andersson1993} that a necessary condition for smooth or polyhomogenous $\mathscr{I}^{+}$ of the development of asymptotically hyperboloidal initial data is given by the shear free condition (c.f. condition (1.1) in \cite{Andersson1993}). It can be verified using \cite{Andersson2002} that asymptotically hyperboloidal initial data sets as defined in Definition \ref{Def::AHyperboloidal} do not satisfy this condition in general. Therefore, the energy and linear momentum loss computed here extend the already known results proved for smooth and polyhomogeneous $\mathscr{I}^{+}$.
\end{remark}

\section{Generalized observers}\label{sec6}

In this section, we explore generalizations of the height functions introduced in Section \ref{sec4} beyond the standard hyperboloid. Our goal is to extend the definitions of energy and linear momentum to cases where the background initial data deviates from the standard hyperboloid at suitably lower orders. 

\begin{definition}\label{Def::generalizedhyperboloid}
    We say that an initial data set $\left(\tilde{M},\tilde{b},\tilde{K}\right)$ is a generalized hyperboloidal initial data set of order $\lambda$ if $\tilde{b} = \left(1 - H(r)^{2}\right)dr^{2} + r^{2}d\sigma^{2}$, for some $H(r)$ satisfying (\ref{choice_H}), and 
    \begin{equation}
        \tilde{K}_{rr}=\frac{H'}{\sqrt{1-H^2}}\tilde{b}_{rr}, \quad \tilde{K}_{\alpha l}=\frac{H}{r\sqrt{1-H^2}} \tilde{b}_{\alpha l},
    \end{equation}
    where $\lambda\geq 3$ is the parameter appearing in the expansion of $H(r)$ in Equation \eqref{choice_H}.
\end{definition}

In the literature, these are sometimes called hyperboloidal data, which is a broader definition than the one we consider in our paper. It is easily checked that choosing $H$ corresponding to the standard hyperboloid we recover $(\Hyp^3,b,b)$. To avoid confusion, we treat the two definitions separately and reserve ``hyperboloid" and ``asymptotically hyperboloidal" only for the standard hyperboloid $(\Hyp^3,b,b)$. It is easy to check that the conditions on $\tilde{K}$ are so that the initial data in Definition \ref{Def::generalizedhyperboloid} embeds as a $\tau$-constant hypersurface in $\left(\mathbb{R}^{3,1},\eta\right)$, as seen from equation (\ref{metric_tau}). The time and spatial translation symmetries of the spacetime correspond to KIDs on the initial data. In vacuum, the KIDs coincide precisely with the normal and tangential components of the decomposition of the Killing vector fields along the initial data \cite{NewKIDs} and therefore we can compute them directly from (\ref{metric_tau}). 
For each choice of $h(r)$, we derive the lapse and shift governing the evolution of the initial data set in the standard Killing time direction in Minkowski.

\begin{lemma}\label{choice_lapseshift}
    Let $\left(\tilde{M},\tilde{b},\tilde{K}\right)$ be a generalized hyperboloidal initial data of order $\lambda\geq 3$, for some $H(r)$ satisfying (\ref{choice_H}). The lapse and shift describing the evolution of the initial data with respect to the standard Killing time in Minkowski are 
    \begin{equation}
        N = \frac{1}{\sqrt{1-H(r)^2}} \quad \text{and} \quad X^r = -\frac{H(r)}{1-H(r)^2},
    \end{equation}
\end{lemma}
\begin{proof}
   This follows by direct computation.
\end{proof}

We define generalized asymptotically hyperboloidal initial data sets that are asymptotic to this generalized background, extending Definition \ref{Def::AHyperboloidal}. 
\begin{definition}\label{Def::GeneralizedAH_IDS}
We say that $\idsmatter$ is a generalized asymptotically hyperboloidal initial data set of order $\lambda\geq 3$ w.r.t. a diffeomorphism at infinity $\Psi$ if there exists a generalized hyperboloidal initial data $(\tilde{M},\tilde{b},\tilde{K})$ of order $\lambda$ and two compact sets $B\subset M$ and $B_0\subset \tilde{M}$ such that 
\[
    \Psi:\tilde{M}\setminus B_0\rightarrow M\setminus B
\]

and $\dot{g}:=\Psi^*g-\tilde{b}$, $p:=\Psi^*K-\tilde{K}-\dot{g}$
satisfy the decay conditions as in Definition \ref{Def::AHyperboloidal}.
\end{definition}

It should be noted here that the generalized background of a generalized asymptotically hyperboloidal initial data is not unique. In particular, two different choices of the function $H$ having the same expansion up to order $r^{-7}$ would give the same asymptotic expansion of the correction terms\footnote{The other components of $\dot{g}$ and $p$ do not depend on the function $H$ so they do not need to be considered here.}
\begin{align*}
    \dot{g}_{rr}=\Psi^*g_{rr}-(1-&H^2),\\
p_{rr}=\Psi^*K_{rr}-\frac{H'}{(1-H^2)^{\frac{3}{2}}}&\tilde{b}_{rr}-\dot{g}_{rr},
\end{align*}
up to the order that contributes in the computation of charges and their evolution. Given a generalized asymptotically hyperboloidal initial data set $\idsmatter$ that is asymptotic via a diffeomorphism at infinity to a generalized hyperboloidal data $(\tilde{M},\tilde{b},\tilde{K})$ defined via a function $\bar{H}$, we define the class of generalized background corresponding to $\idsmatter$ as the generalized hyperboloidal data defined via $H$ in 
\begin{equation}\label{Eqn::classofH}
    \mathcal{H}:=\left\{H(r)=1-\frac{1}{2r^2}+O_2(r^{-3}) \ |\ H-\bar{H} \sim O_{2}(r^{-7})\right\}.
\end{equation}

In the following Theorem, we extend the definition of energy and linear momentum to generalized asymptotically hyperboloidal initial data sets as in Definition \ref{Def::GeneralizedAH_IDS} and prove that the resulting expressions are the same as in Theorem \ref{THM::Energy_and_Linearmomentum_AH}. These charges will not depend on the choice of representative in the class $\mathcal{H}$. More precisely, given generalized asymptotically hyperboloidal initial data, there exist coordinates at infinity such that the geometry of the initial data approaches, in a suitable way, a class of generalized hyperboloidal data. After selecting a representative generalized background in this class, which is a hypersurface in Minkowski spacetime, we compute the KIDs corresponding to time and spatial translations and obtain the following result.  

\begin{theorem}
    Let $\idsmatter$ be a generalized asymptotically hyperboloidal initial data set of order $\lambda\geq 3$, in the sense of Definition \ref{Def::GeneralizedAH_IDS}.
    The energy and linear momentum of this initial data are defined as the Michel charges corresponding to the KIDs of the generalized hyperboloidal initial data $(\tilde{M},\tilde{b},\tilde{K})$, associated with the time and spatial translations in Minkowski. They are well defined and have the following expressions
    \begin{equation}\label{Eqn::EnergyAH::generalized}
        E = \frac{1}{16\pi}\int_{\Sph^{2}} 2\mathbf{m}_{rr} + 3 \, \trace_{\sigma}\ \!^{g}\mathbf{m}+ 2\trace_{\sigma}\ \!^{k}\mathbf{m} \: dA_{\Sph^{2}}
    \end{equation}
\begin{equation}\label{Eqn::MomentumAH::generalized}
    P^{i} = \frac{1}{16\pi}\int_{\Sph^{2}}\left( 2\mathbf{m}_{rr} + 3 \, \trace_{\sigma}\ \!^{g} \mathbf{m} + 2\trace_{\sigma}\ \!^{k} \mathbf{m} \right) x^{i} \, dA_{\Sph^{2}}, \quad i = 1,2,3.
\end{equation}
where $x^{i}$ denote the first spherical harmonics on the unit sphere $\Sph^2$. In particular, energy and linear momentum of generalized asymptotically hyperboloidal initial data are well defined and are independent of the choice of generalized background or, equivalently, of the choice of $H\in\mathcal{H}$.
\end{theorem}
\begin{proof}
    Let $(\tilde{M},\tilde{b},\tilde{K})$ be a representative of the class of generalized backgrounds corresponding to the initial data $\idsmatter$, with respect to a diffeomorphism at infinity $\Psi$. Let $H$ be the function defining this background, as in Definition \ref{Def::generalizedhyperboloid}. Then, the KID $(V, Y=\alpha^\#)$  associated with the time translation symmetry in Minkowski (Lemma \ref{choice_lapseshift}) can be written in terms of $H$:
    \[
        V=\frac{1}{\sqrt{1-H^2}},\quad Y^r=\frac{-H}{1-H^2}.
    \]
    The energy of $\idsmatter$ is
    \allowdisplaybreaks[0]
    \begin{align*}
        E=\lim_{r\to\infty}\frac{1}{16\pi}\int_{\Sph^2_r}\Big\{V(\divergence_0\dot{g}-d\trace_0\dot{g})-\iota_{\nabla V}\dot{g}+(\trace_0\dot{g})dV+2\big(\iota_\alpha\dot{K}-(\trace_0\dot{K})\alpha\big)\\
        +(\trace_0\dot{g})\iota_\alpha \tilde{K}+\langle \tilde{K},\dot{g}\rangle_0\alpha -2\iota_\alpha(\dot{g}\circ \tilde{K})\Big\} \; \frac{\partial_r^{\mbox{\footnotesize{tang}}}}{|\partial_r^{\mbox{\footnotesize{tang}}}|_{\eta}} \; dA_{\Sph^2_r},
    \end{align*}
    whenever the integral is well-defined and the limit exists. Traces and derivatives in the above expression are taken with respect to the generalized background metric $\tilde{b}$. The vector field $\partial_r^{\mbox{\footnotesize{tang}}}$, is normal to the round spheres $\Sph^2_r$ on $(\tilde{M},\tilde{b},\tilde{K})$ and tangential to the initial data. In terms of the standard frame $\partial_t,\partial_r,\partial_\alpha$ in Minkowski:
    \begin{equation*}
        \partial_r^{\mbox{\footnotesize{tang}}}=-H^2\sqrt{1-H^2}\partial_t+\frac{1}{H\sqrt{1-H^2}}\partial_r
    \end{equation*}
    and, by construction, $dr(\partial_r^{\mbox{\footnotesize{tang}}})=1$. 
    
    The Christoffel symbols of the generalized metric are the same as for the standard hyperboloid, except for the following
    \begin{equation*}
        \Gamma^r_{rr}=\frac{-H H'}{1-H^2},\quad
        \Gamma^r_{\alpha\beta}=\frac{-r}{1-H^2}\sigma_{\alpha\beta}.
    \end{equation*}
Substituting the above in the expression for the energy, leads to the following charge integrand:
    \begin{align*}
        \frac{1}{(1-H^2)}\Big[\frac{1}{1-H^2}\Big(\dot{g}_{rr,r}+\frac{2H H'}{1-H^2}\dot{g}_{rr}\Big)+\frac{\sigma^{\alpha\beta}}{r^2}\Big(\dot{g}_{\alpha r,\beta}+\frac{r\sigma_{\alpha\beta}}{1-H^2}\dot{g}_{rr}-\ \!^{\Sph^2}\Gamma^\gamma_{\alpha\beta}\dot{g}_{r\gamma}-\frac{1}{r}\dot{g}_{\alpha\beta}\Big)\\
        -\frac{d}{dr}\left(\frac{1}{1-H^2}\right)\dot{g}_{rr}-\frac{1}{1-H^2}\dot{g}_{rr,r}+\frac{2}{r^3}\trace_\sigma(\dot{g}_{\alpha\beta})-\frac{1}{r^2}\trace_\sigma(\dot{g}_{\alpha\beta,r})\Big] -\frac{H H'}{(1-H^2)^3}\dot{g}_{rr}\\
        +\left(\frac{1}{1-H^2}\dot{g}_{rr}+\frac{1}{r^2}\trace_\sigma(\dot{g}_{\alpha\beta})\right)\frac{H H'}{(1-H^2)^2}+\frac{2}{\sqrt{1-H^{2}}}\left[\frac{-H}{1-H^2}\dot{K}_{rr}+(\trace_{\tilde{b}}\dot{K})\tilde{b}_{rr}\frac{H}{1-H^2}\right]\\
        -\left(\frac{1}{1-H^2}\dot{g}_{rr}+\frac{1}{r^2}\trace_\sigma(\dot{g}_{\alpha\beta})\right)\frac{H H'}{(1-H^2)^{2}}-\left[\frac{H H'}{(1-H^2)^3}\dot{g}_{rr}+\frac{H^2}{r^3(1-H^2)}\trace_\sigma(\dot{g}_{\alpha\beta})\right]\\ +\frac{2H H'}{(1-H^2)^3}\dot{g}_{rr}.
    \end{align*}
  The expression for energy of a generalized asymptotically hyperboloidal initial data is obtained computing the coefficients of the contributing terms,
    \begin{equation*}
        \left(\frac{2}{r^3(1-H^2)}-\frac{H^2}{r^3(1-H^2)}+\frac{2H}{r^2\sqrt{1-H^{2}}}\right)\trace_\sigma(\dot{g}_{\alpha\beta}) = \frac{3}{r^{2}}\trace_{\sigma}{}^{g}\mathbf{m}  + O\left(r^{-3}\right),
    \end{equation*}

    \begin{equation*}
        \frac{2H}{\sqrt{1-H^{2}}}\frac{1}{r^{2}} \trace_{\sigma}(p_{\alpha\beta}) = \frac{2}{r^{2}}\trace_{\sigma}{}^{k}\mathbf{m} + O\left(r^{-3}\right),
    \end{equation*}
and
    \begin{equation*}
        \frac{2}{r(1-H^{2})^{2}}\dot{g}_{rr} = \frac{2}{r^{3}} \mathbf{m}_{rr} + O\left(r^{-3}\right)
    \end{equation*}
Therefore, it is clear that we retrieve the expressions for energy seen in Theorem \ref{THM::Energy_and_Linearmomentum_AH}. The expression for linear momentum is derived analogously.  
\end{proof}

\begin{remark} 
 It should be noted that the above can be further generalized to include hyperboloids of radius different from 1. In this case, Equation \eqref{choice_H} becomes $H(r)=1-\frac{C}{r^2}+O(r^{-\lambda})$, for some constant $C>0$. However, defining generalized asymptotically hyperboloidal data as in Definition \ref{Def::GeneralizedAH_IDS} with such a background would require careful adjustments to the expression of the asymptotic decay of the data, as a dependence on the constant $C$ would then emerge. Since we do not expect this generalization to yield significant additional insights, we do not pursue it here. 
 The same procedure can be extended to produce (generalized) hyperboloidal foliations in the Schwarzschild spacetime \cite{Anil}. In this case, the height function is given by \begin{equation} 
    h(r) = r + 2m \ln(r) + \frac{C}{r} + O_{3}\left(r^{-\lambda+1}\right),    
    \end{equation} 
    where $m$ is the Schwarzschild mass. The same holds for the Kerr spacetime. In general, this procedure to construct height functions extends to asymptotically Minkowskian spacetimes. 
\end{remark}

\section{Discussion}\label{discussion}
In this paper, we derived the flux of energy and linear momentum for asymptotically hyperboloidal initial data sets using the Einstein evolution equations. Our approach, based on geometric invariants of the initial data, provides a different perspective compared to traditional frameworks like the Bondi-Sachs formalism.
In the future, we aim to investigate the evolution of energy and linear momentum with respect to the generalized observers described in Section \ref{sec6}. We expect that the class of observers would need to be suitably restricted to ensure that E-P chargeability is preserved. For this class of observers, we anticipate that the formulae for energy and linear momentum loss will be the same as those obtained in this paper. An important direction for future work is to compare our framework with the Hamiltonian charges of energy and linear momentum and to explore whether any discrepancies emerge. This is particularly motivated by observations of similar issues at the asymptotically Euclidean level, where angular momentum has been shown to exhibit related behavior \footnote{Personal communication with Carla Cederbaum, December 2024}.

Additionally, we plan to extend this study to other conserved quantities, such as angular momentum and the center of mass, where supertranslation ambiguities may arise. A further objective is to explore initial data with polyhomogeneous expansions \cite{Chrusciel:1993hx,Andersson1993}. Finally, we aim to apply this framework to models of matter, such as electromagnetism, to study the flux of conserved quantities in the presence of matter fields.

\begin{appendices}
 \section{Christoffel Symbols of the Standard Hyperboloid $(\mathbb{H}^{3}, b)$}\label{appA}
 Consider the hyperboloid $\mathbb{H}^{3}$ with the metric written in spherical coordinates as
\begin{equation}
    b = \frac{dr^{2}}{1+r^{2}} + r^2\sigma,
\end{equation}
where $\sigma$ is the metric on the standard 2-sphere $\mathbb{S}^{2}$. We denote the Christoffel symbols corresponding to the hyperbolic metric as ${}^{b}\Gamma^{i}_{jk}$. We use Latin indices to denote the coordinates on the hyperboloid, and Greek indices to denote the coordinates on the 2-sphere. The Christoffel symbols are computed using the standard formula
\begin{equation}
    {}^{b}\Gamma^{i}_{jk} = \frac{1}{2}b^{ip}\left(b_{pj,k} + b_{pk,j} - b_{jk,p}\right).
\end{equation}
The non-vanishing Christoffel symbols of $b$ are as follows:
    $${}^{b}\Gamma^{r}_{rr} = \frac{-r}{1+r^{2}},$$
    $${}^{b}\Gamma^{\alpha}_{r\beta} = \frac{1}{r} \delta^{\alpha}_{\beta}, $$
    $${}^{b}\Gamma^{r}_{\alpha\beta} = -r(1+r^{2})\sigma_{\alpha\beta},$$
    $${}^{b}\Gamma^{\alpha}_{\beta\gamma} = {}^{\mathbb{S}^{2}}\Gamma^{\alpha}_{\beta\gamma}.$$

\section{Christoffel Symbols of Asymptotically Hyperboloidal Initial Data Sets}\label{appB}
In this section, we consider an asymptotically hyperboloidal metric as seen in Definition \ref{Def::AHyperboloidal}. We start by computing the inverse metric and inverse extrinsic curvature components:
\begin{equation}\label{AHyperboloidal_inv}
    \begin{aligned}
        \dot{g}^{rr} &= -\frac{\mathbf{m}_{rr}}{r} - \frac{\tilde{\mathbf{m}}_{rr}}{r^{2}} + O_{2}\left(r^{-3}\right),  \quad  
        \dot{g}^{\alpha r} = - \frac{\mathbf{g}^{\alpha}_{r}}{r^{3}} + O_{2}(r^{-4}),  \\ 
        \dot{g}^{\alpha\beta} &= - \frac{{}^{0}\mathbf{g}^{\alpha\beta}}{r^{4}} - \frac{{}^g\mathbf{m}^{\alpha\beta}}{r^{5}} 
        - \frac{{}^{g}\tilde{\mathbf{g}}^{\alpha\beta}}{r^{6}} +O_{2}\left(r^{-7}\right),
    \end{aligned}
\end{equation}
and
\begin{equation}
    \begin{aligned}
        {p}^{rr} &= - \frac{\mathbf{p}_{rr}}{r} - \frac{\tilde{\mathbf{p}}_{rr}}{r^{2}} + O_{2}\left(r^{-3}\right),  \quad  
        {p}^{\alpha r} = -\frac{\mathbf{p}_{r}^{\alpha}}{r^{3}} + O_{2}(r^{-4}),  \\  
        {p}^{\alpha\beta} &= -\frac{{}^{0}\mathbf{p}^{\alpha\beta}}{r^{4}} -\frac{\!^k\mathbf{m}^{\alpha\beta}}{r^{5}} 
        - \frac{{}^{k}\tilde{\mathbf{m}}^{\alpha\beta}}{r^6} + O_{2}\left(r^{-7}\right).
    \end{aligned}
\end{equation}
We denote the Christoffel symbols corresponding to the asymptotically hyperboloidal  metric as $\Gamma^{i}_{jk}$. We express here only the first two subleading terms beyond the background Christoffel symbols, which are the highest-order corrections we consider.
\begin{align*}
    \Gamma^{r}_{rr} &= {}^{b}\Gamma^{r}_{rr} + \frac{1}{2}\dot{g}^{rr}\left(b_{rr,r}\right) + \frac{1}{2}b^{rr}\left(\dot{g}_{rr,r}\right) + O_{1}(r^{-7})
    \\&= {}^{b}\Gamma^{r}_{rr} -\frac{3}{2}\frac{\mathbf{m}_{rr}}{r^{4}}  -\frac{2\tilde{\mathbf{m}}_{rr}}{r^{5}} + O_{1}(r^{-6}),
\end{align*}
\begin{align*}
    \Gamma^{\alpha}_{r\beta} &= {}^{b}\Gamma^{\alpha}_{r\beta} + \frac{1}{2}\dot{g}^{\alpha\delta}\left(b_{\delta\beta,r}\right) + \frac{1}{2}b^{\alpha\delta}\left(\dot{g}_{\delta r,\beta} + \dot{g}_{\delta \beta,r} - \dot{g}_{r\beta,\delta}\right) + O_{1}(r^{-7})
    \\&= {}^{b}\Gamma^{\alpha}_{r\beta} -\frac{3}{2}\frac{\mathbf{m}^{\alpha}_{\beta}}{r^{4}} - \frac{2\tilde{\mathbf{m}}^{\alpha}_{\beta}}{r^{5}} + \frac{1}{2} \left(\frac{\mathbf{g}^{\alpha}_{r,\beta}}{r^{5}} - \frac{\mathbf{g}_{\beta r},^{\alpha}}{r^{5}}\right) + O_{1}\left(r^{-6}\right),
\end{align*}
\begin{align*}
    \Gamma^{r}_{\alpha\beta} &= {}^{b}\Gamma^{r}_{\alpha\beta} 
    + \frac{1}{2} \dot{g}^{rr} \left(-b_{\alpha\beta,r}\right)  
    + \frac{1}{2} \dot{g}^{r\rho} \left( b_{\rho\alpha,\beta} 
    + b_{\rho\beta,\alpha} - b_{\alpha\beta,\rho} \right)  \\
    &\quad + \frac{1}{2} b^{rr} \left(\dot{g}_{r\alpha,\beta} 
    + \dot{g}_{r\beta,\alpha} - \dot{g}_{\alpha\beta,r} \right) 
    + O(r^{-3}) \\
    &= {}^{b}\Gamma^{r}_{\alpha\beta} + \mathbf{m}_{rr} \sigma_{\alpha\beta} 
    + \frac{1}{2} {}^{g} \mathbf{m}_{\alpha\beta} 
    + \frac{\tilde{\mathbf{m}}_{rr}}{r} \sigma_{\alpha\beta} 
    + \frac{1}{2} \left( \frac{{}^{g} \mathbf{g}_{r\alpha,\beta} 
    + {}^{g} \mathbf{g}_{r\beta,\alpha}}{r} \right) \\
    &\quad - {}^{\Sph^{2}} \Gamma^{\rho}_{\alpha\beta} 
    \frac{\mathbf{g}_{r\rho}}{r} 
    + \frac{{}^{g} \tilde{\mathbf{m}}_{\alpha \beta}}{r} 
    + O_{1} (r^{-2}),
\end{align*}

\begin{align*}      
    \Gamma_{\alpha\beta}^{\delta} &= {}^{b}\Gamma^{\delta}_{\alpha\beta} + \frac{1}{2}\dot{g}^{\delta\rho}\left(b_{\rho\alpha,\beta} + b_{\rho\beta,\alpha} - b_{\alpha\beta,\rho}\right) + \frac{1}{2}b^{\delta\rho}\left(\dot{g}_{\rho\alpha,\beta} + \dot{g}_{\rho\beta,\alpha} - \dot{g}_{\alpha\beta,\rho}\right) + O_{1}(r^{-5})
    \\&= {}^{\mathbb{S}^{2}}\Gamma^{\delta}_{\alpha\beta} +  O(r^{-3}),
\end{align*}
\begin{align*}
    \Gamma^{r}_{\alpha r} &= \frac{1}{2}\dot{g}^{rr}\left( b_{rr,\alpha}\right) + \frac{1}{2}b^{rr}\left(\dot{g}_{rr,\alpha}\right) + \frac{1}{2}\dot{g}^{r\rho}\left(b_{\rho \alpha ,r}\right) + O_{1}\left(r^{-3}\right) 
    \\& = -\frac{\mathbf{g}_{r\alpha}}{r^{2}} + O_{1}\left(r^{-3}\right).
\end{align*}
For reference, we also list here the products of certain Christoffel symbols.
\begin{align*}
    \Gamma^{r}_{rr} \Gamma^{r}_{\alpha\beta} &= {}^{b} \Gamma^{r}_{rr} \, {}^{b} \Gamma^{r}_{\alpha\beta}  
    + \frac{1}{2} \frac{\mathbf{m}_{rr}}{r} \sigma_{\alpha\beta}  
    - \frac{1}{2} \frac{{}^{g} \mathbf{m}_{\alpha\beta}}{r}  
    + \frac{\tilde{\mathbf{m}}_{rr}}{r^{2}} \sigma_{\alpha\beta}  \\
    &\quad +\frac{1}{2} {}^{\Sph^{2}} \Gamma_{\alpha\beta}^{\rho} 
    \frac{\mathbf{g}_{r\rho}}{r^{2}}  
    - \frac{1}{2} \left( \frac{ \mathbf{g}_{r\alpha,\beta} 
    +  \mathbf{g}_{r\beta,\alpha}}{r^{2}} \right)  
    - \frac{\tilde{{}^{g} \mathbf{m}}_{\alpha\beta}}{r^{2}}  
    + O_{1} (r^{-4}),
\end{align*}

\begin{align*}
    \Gamma^{\rho}_{\alpha r} \Gamma^{r}_{\beta\rho} &= {}^{b} \Gamma^{\rho}_{\alpha r} \, {}^{b} \Gamma^{r}_{\beta\rho}  
    + \frac{1}{r} \mathbf{m}_{rr} \sigma_{\alpha\beta}  
    + 2 \frac{{}^{g} \mathbf{m}_{\alpha\beta}}{r}  
    + \frac{\tilde{\mathbf{m}}_{rr}}{r^{2}} \sigma_{\alpha\beta}  
    + 3 \frac{{}^{g} \tilde{\mathbf{m}}_{\alpha\beta}}{r^{2}}  \\
    &\quad + \frac{\mathbf{g}_{r\alpha,\beta}}{r^{2}}  
    - {}^{\Sph^{2}} \Gamma^{\rho}_{\alpha\beta} \frac{\mathbf{g}_{r\rho}}{r^{2}}  
    + O_{1} \left( r^{-3} \right),
\end{align*}
\begin{align*}
    \Gamma^{\delta}_{\delta r} \Gamma^{r}_{\alpha\beta} &= {}^{b} \Gamma^{\delta}_{\delta r} \, {}^{b} \Gamma^{r}_{\alpha\beta}  
    + 2 \frac{\mathbf{m}_{rr}}{r} \sigma_{\alpha\beta}  
    + \frac{1}{r} \mathbf{m}_{\alpha\beta}  
    + \frac{3}{2} \frac{\trace_{\sigma} \mathbf{m}}{r} \sigma_{\alpha\beta}  
    + 2 \frac{\tilde{\mathbf{m}}_{rr}}{r^{2}} \sigma_{\alpha\beta}  \\
    &\quad  - {}^{\Sph^{2}} \Gamma^{\rho}_{\alpha\beta} \frac{\mathbf{m}_{r\rho}}{r^{2}}  
    + \frac{\mathbf{m}_{r\alpha,\beta} + \mathbf{m}_{r\beta,\alpha}}{r^{2}}  
    + 2 \frac{{}^{g} \tilde{\mathbf{m}}_{\alpha\beta}}{r^{2}}  
    + 2 \frac{\trace_{\sigma} {}^{g} \tilde{\mathbf{m}}}{r^{2}} \sigma_{\alpha\beta}  
    + O_{1} \left( r^{-3} \right),
\end{align*}
\begin{align*}
    \Gamma_{\rho r}^{r}\Gamma_{\alpha\beta}^{\rho} = -\frac{\mathbf{g}_{r\alpha}}{r^{2}}{}^{\Sph
    ^{2}}\Gamma_{\alpha\beta}^{\rho} + O_{1}\left(r^{-3}\right).
\end{align*}

\section{Evolution of hyperboloidal initial data}\label{App:evolution}
Consider the standard hyperboloidal initial data set $(\Hyp^{3},b,b)$, with
\begin{equation*}
    b=\frac{1}{1+r^2}dr^2+d\sigma^2.
\end{equation*}
We consider an evolution of the above initial data set with respect to its associated lapse and shift given by Remark \ref{rmk:standardhyp}. The evolution of the metric $b$ with respect to this class of observers is given by

\begin{equation*}
\begin{aligned}
    \frac{\partial}{\partial\tau}b_{rr} &= 2Nb_{rr} + \mathcal{L}_{X}b_{rr}
    \\& = \frac{2}{\sqrt{1+r^{2}}} - r\sqrt{1+r^{2}}\left(\frac{-2r}{(1+r^{2})^{2}}\right) - \frac{2r^{2}}{(1+r^{2})^{\frac{3}{2}}} - \frac{2}{\sqrt{1+r^{2}}}
    \\& = 0,
\end{aligned}
\end{equation*}
and
\begin{equation*}
    \begin{aligned}
        \frac{\partial}{\partial\tau}b_{\alpha\beta} &= 2Nb_{\alpha\beta} + \mathcal{L}_{X}b_{\alpha\beta} 
        \\& = 2r^{2}\sqrt{1+r^{2}}\sigma_{\alpha\beta} - r\sqrt{1+r^{2}} \cdot2r = 0.
    \end{aligned}
\end{equation*}

The evolution of the second fundamental form $K=b$ also can be seen to give that $\frac{\partial}{\partial\tau}K = 0$, using that $\Ric = -2b$.

\end{appendices}

\bibliography{sn-bibliography}

\end{document}